\documentclass[journal,draftcls,onecolumn,12pt]{IEEEtranTCOM}
\usepackage{amsmath,graphicx,amsthm}
\usepackage{amsfonts,amssymb,dsfont}
\usepackage{algorithmicx}
\usepackage{algorithm}
\usepackage{caption}
\usepackage{cite}
\usepackage{float}
\usepackage{cleveref}
\usepackage{lettrine}
\usepackage{mathtools}
\usepackage{mathrsfs}
\usepackage{bm}
\usepackage{latexsym}
\usepackage{tikz}
\usepackage{multirow} 
\usepackage{bbm}
\usepackage{multirow}
\usepackage[version=4]{mhchem} 
\usepackage{subfig}

\usepackage{graphicx}
\graphicspath{ {./png2files/}}

\DeclareGraphicsExtensions{.jpg,pdf,.png,.bmp,.eps}
\theoremstyle{plain}
\newtheorem{lem}{Lemma}

\newtheorem{theorem}{Theorem}
\theoremstyle{definition}

\newtheorem{example}{Example}

\title{An Analytical Model for Molecular Communication over a Non-linear Reaction-Diffusion Medium}
 
\author{\IEEEauthorblockN{Hamidreza Abin,
Amin Gohari, and
Masoumeh Nasiri-Kenari}\\
\IEEEauthorblockA{\IEEEauthorrefmark{1}Department of Electrical Engineering, Sharif University of Technology}}

\begin{document}

\maketitle
\vspace{-1cm}
\begin{abstract}
One of the main challenges in diffusion-based molecular communication is dealing with the non-linearity of reaction-diffusion chemical equations. While numerical methods can be used to solve these equations, a change in the input signals or the parameters of the medium requires one to redo the simulations. This makes it difficult to design modulation schemes and practically impossible to prove the optimality of a given transmission strategy. In this paper, we provide an analytical technique for modeling the  non-linearity of chemical reaction equations based on the perturbation method.  The perturbation method expresses the solution in terms of an infinite power series. An approximate solution can be found by keeping the leading terms of the power series. The approximate solution is shown to track the true solution if either the simulation time interval or the reaction rate is sufficiently small. Approximate solutions for long time intervals are also discussed. An illustrative example is given. For this example, it is shown that when the reaction rate (or the total time interval) is low,  instead of using a continuous release
waveform,  it is optimal for the transmitters to release molecules at two time instances.
\footnote{A short version of this paper appeared 
at \cite{abin2020molecular}. The differences between the conference version and the journal version are as follows:  Sections II.A,  II.C, II.D, V, Theorem 1, Proof of Theorem 2, Appendix B, Appendix C, Appendix D, Figures 2, 3,  4, 8, 9, 10, 12, 13, 14, 15 are missing from the conference version (and the simulation section is expanded). Appendix A is also expanded compared to the conference version. The conference version (6-pages) is less than half the length of the journal version.}
\end{abstract}


\IEEEpeerreviewmaketitle


\section{Introduction}
Molecular communication (MC) uses molecules as carriers of information. Molecular communication 
is envisioned to be applicable in a wide range of engineering and medical applications, especially as a means to realize communications among  engineered biological nanomachines
 (see \cite{pierobon2010physical, nakano2012molecular, guo2015molecular, survey:medicine}). In diffusion-based MC, molecules are released into the medium by molecular transmitters. Information is coded by the transmitter in the type, number, or release time of the molecules. The released molecules randomly diffuse in the medium, with some reaching the molecular receivers. The receiver decodes the message sent by the transmitter based on the number of sensed molecules (see  \cite{kuran2011modulation, arjmandi2013diffusion, srinivas2012molecular} for some  examples of modulation schemes).

Broadly speaking, two main types of diffusion-based MC exists: microscale MC and macroscale MC. In microscale MC, a small number of molecules are released into the medium by transmitters. The movement of each molecule is random and follows Brownian motion. On the other hand, in macroscale MC, a large number of molecules are released into the medium (the number of released molecules is in the order of moles). Macroscale MC is of interest as an alternative to electromagnetic wave-based systems in certain media \cite{guo2015molecular, farsad2016comprehensive}. 
In macroscale MC, instead of focusing on the location of individual molecules, by the law of large numbers, the \emph{average concentration} of molecules is described by Fick's law of diffusion, which is a \emph{deterministic} differential equation. Everything is described by deterministic equations until this point. However, various sources of noise or imperfection could be considered for molecular transmitters and receivers \cite{farsad2016comprehensive, gohari2016information}. In particular, 
the literature on MC usually associates a measurement noise to the receivers. The variance of this measurement noise depends on the concentration level of molecules at the receiver  \cite{einolghozati2012collective}. In particular, a commonly used measurement noise is the  ``Poisson particle counting noise" where the measurement follows a Poisson distribution whose mean is proportional to the concentration of molecules at the time of measurement.

\textbf{Motivation:} One of the unique features of MC with no parallel in classical wireless communication is \emph{chemical reactions}: different types of molecules can react with each other and form new types of molecules in the medium. This feature of MC poses a challenge in macroscale MC since equations describing chemical reaction with diffusing particles are \emph{nonlinear} partial differential equations with no closed-form solutions\cite{noclosedform:1,noclosedform:2,noclosedform:3}. 
In other words, there is no analytical closed-form formula that describes the variations in the concentration of molecules over time and space for given boundary conditions, inputs to the medium, and system parameters such as the diffusion and reaction rate constants. These equations are commonly solved by numerical methods. While numerical techniques (such as the finite difference method and the finite element method) can provide accurate solutions, they require extensive computational resources. More importantly, these solutions are not generalizable: assume that we have solved the reaction-diffusion equations for one particular release pattern of molecules by the transmitters into the medium. Next, if we make some changes to the release pattern of molecules, we should start all over again and solve the reaction-diffusion equations again for the new release pattern. Numerical methods provide little insight into how the solution would change if the input or parameters of the medium change. This has a major drawback for the design of coding and modulation schemes, where we wish to optimize and find the best possible release pattern. A tractable (even if approximate) model for the solution of the reaction-diffusion systems is required to design codes or modulation schemes. Only with such a model can one formally the optimality of certain waveforms for a given communication problem.

\textbf{Our contribution:} 
We apply the perturbation theory to  design modulation schemes among manufactured devices in macroscale MC systems involving chemical reactions. 
Perturbation theory is a set of 
mathematical methods for finding approximate solutions. It was originally proposed to calculate the motions of planets in the solar system, which is an intractable problem (as the so-called three-body problem shows). The theory has been broadly applied in different scientific areas such as physics and mechanics for solving algebraic equations with no analytical solutions. 
While the perturbation theory has been used in the study of differential equations in general and 
in the context of fiber optics communication systems \cite{pert-opt-cite}, to the best of our knowledge it has not been hitherto utilized in molecular communication.\footnote{ Note that  \cite{singularpert} uses the idea of \emph{
singular perturbation} which should not be confused with the \emph{perturbation theory}.} Just as non-linear functions can be locally approximated by the first terms of their Taylor series, perturbation theory shows that non-linear reaction-diffusion equations can be also approximated by their lower order terms.   Our main contribution is a proposal to design coding and modulation schemes based on the lower order terms, which admit closed-form solutions. In other words, we provide an analytically tractable (approximate) model for studying chemical reactions in molecular communication. The accuracy of the model can be increasingly improved by including more terms from the Taylor series.

We concretely illustrate the perturbation method by considering the following reaction-diffusion system in which molecules of type $\mathtt{A}$ and $\mathtt{B}$ react with each other in the medium and produce molecules of type $\mathtt{C}$:
\begin{align}\label{ex1:reactionN1}
\mathtt{A}  + \mathtt{B} 
\underset{\lambda_2}{\stackrel{\lambda_1}{\rightleftharpoons}} \mathtt{C}.
\end{align}
Forward and backward reaction rates $\lambda_1$ and $\lambda_2$ represent the speed at which the chemical reaction takes place in the forward and backward directions. This reaction-diffusion system is chosen because it is one of the simplest chemical reactions with no closed-form solution. This basic  reaction  appears in many contexts, e.g. $\text{base}+\text{acid}\rightarrow \text{salt}$, or $NH_3+{H_2O}\rightarrow {NH_4OH}$,  $\text{ADP}+\text{P}\rightarrow \text{ATP}$, or
$
\text{lipoprotein} + \text{PRAP1} 
\underset{\lambda_2}{\stackrel{\lambda_1}{\rightleftharpoons}}\text{PRAP1-lipoprotein}$ 
in biology. 
It is used in \cite{farahnak2018medium:j:16} to design a communication strategy between engineered devices. It may be also used in chemical computing, \emph{e.g.} to implement a binary AND gate (as the product molecules are produced only when both reactants are present).

The application considered in this paper is a communication system involving three engineered devices, two transmitters, and one receiver: the two transmitters release molecules of type $\mathtt A$ and $\mathtt B$ respectively, and the receiver senses the product molecules of type $\mathtt C$. 

The goal is to compute the best transmission waveform for each transmitter so that the receiver has minimum error probability to decode the messages of both transmitters (a multiple access channel setting). 
In this setting, we find the optimal waveforms that minimize the error probability in the low rate reaction regime (or in the small transmission time interval regime). In particular, we show that instead of using a continuous release waveform, it is optimal for the transmitters to release molecules at
two time instances.

To illustrate the broad applicability of our technique, we also discuss other reaction-diffusion systems (different from \eqref{ex1:reactionN1})  in the Appendix \ref{ex2,3}. These examples cover scenarios where  the  chemical reaction is either facilitating or impeding effective communication.

\textbf{Related works:} The math literature on nonlinear partial differential equations (PDEs) studies different aspects of these equations, such as the existence and uniqueness of the solution, semianalytical solutions, numerical solutions, etc. Since most PDEs do not have closed-form solutions,
showing the existence and in some cases, uniqueness of solution are important topics in the theory of PDEs. Unlike ordinary differential equations (ODE) which have a straightforward theorem for the existence and uniqueness of a solution, there is no general method for PDEs to show the existence and uniqueness of the solution. For instance, the existence and smoothness of solutions for the  Navier-Stokes equations that describe the motion of a fluid in space are fundamental open problems in physics. 
Semianalytical techniques are series expansion methods in which the solution is expressed as an infinite series of explicit functions.  Semianalytical techniques include  the Adomian decomposition method, Lyapunov artificial small parameter method, homotopy perturbation method, and perturbation methods  \cite{liao2003beyond,liao2012homotopy}.
 There are many numerical methods for solving PDEs which are classified in terms of complexity and stability of the solution. The finite element method, finite difference method, spectral finite element method, meshfree finite element method, discontinuous Galerkin finite element method are some examples of numerical techniques for solving PDEs \cite{numericalmethode:1,noclosedform:2}. For example, in the finite difference method, we partition the domain using a mesh and 
 approximate derivatives with finite differences computed over the mesh. One of the challenges of using this method is determining the appropriate mesh size to guarantee the stability of the solution.\footnote{ Our own experiment with this method indicates that choosing the appropriate mesh size is particularly challenging when we have a reaction-diffusion system involving molecules that have very different diffusion constants (especially when one diffusion constant is more than ten times another diffusion constant).}
 
 There are some prior works in the literature on chemical physics that utilize the perturbation theory. For instance,  the average survival time of  diffusion-influenced chemical reactions  is studied in \cite{pagitsas1992perturbation:e:1:1}.  In \cite{kryvohuz2014nonlinear:e:1:2} some classes of chemical kinetics  (describe by ODEs) have been solved by the perturbation theory. However, the setup and approach in these problems differ from the one encountered in MC. 
 
 Due to their evident importance to MC,  chemical reactions have been the subject of various studies. A molecular communication system involves transmitters, receivers, and the molecular media; chemical  reactions may occur in each of these individual building blocks. Transmitters could use chemical reactions to produce  signaling molecules \cite{bi2019chemical:j:9}. This is a chemical reaction inside the transmitter unit. On the receiver side, one might have ligand (or other types of) receptors on the surface of the receiver which react with incoming molecules. These receptors have been the subject of various studies, e.g., see \cite{chou2015impact:e:2:2:3,Ligand0, Ligand1, Ligand2, Ligand3,kuscu2018modeling:e:2:3:2, ahmadzadeh2016comprehensive:e:2:3:1,chou2014molecular:e:2:2:4}. 
 Finally, chemical reactions have also been considered in the communication medium. This aspect of chemical reactions is of interest in this work. 
 It is pointed out that chemical reactions could be used to suppress signaling molecules and thereby reduce
  intersymbol interference (ISI) and the signal-dependent measurement noise 
  \cite{noclosedform:1,cho2017effective:j:12}, or to amplify the  signal \cite{nakano2011repeater:F:22}. Complicated propagation patterns can be obtained by using chemical reactions\cite{nakano2015molecular:F:19}, and negative molecular signals could be implemented using chemical reactions  \cite{farsad2016molecular:F:16,wang2014transmit:F:20,mosayebi2017type:F:21}. Notably, chemical reactions are shown to be beneficial for coding and
 modulation design \cite{farsad2016molecular:F:16,farahnak2018medium:j:16,noclosedform:2,nakano2015molecular:F:19,wang2014transmit:F:20,mosayebi2017type:F:21}. Since solving the reaction-diffusion equations is  intractable, these works either use numerical simulations to back up the presented ideas or else simplify the reaction-diffusion process via idealized assumptions about chemical reactions. Authors in \cite{cao2019chemical:j:ro} provide an iterative numerical algorithm for solving reaction-diffusion equations.
 In \cite{farsad2017novel:e:2:2:2}, a neural network is used as the decoder in a medium with chemical reactions.
 
 This paper is organized as follows: in Section \ref{sec::generalperturbation}, we illustrate the perturbation method for solving reaction-diffusion equations through an example.
  In Section \ref{Sec:Mod} we design modulation schemes for our example in Section \ref{sec::generalperturbation}. In Section \ref{sec:Val}, we validate our model through numerical simulation. Finally concluding remarks and future work are given in Section \ref{Conclusion and Future Work}. 

 \textbf{Notations and Remarks:}

 For functions $u_i(x,t)$ and $v_i(x,t)$, we define index convolution as follows:
 \begin{align}\label{dicreteconv}
 (u_{0:i-1}*^d v_{0:i-1})(x,t)=\sum_{j=0}^{i-1} u_{j}(x,t)v_{i-1-j}(x,t).
 \end{align}
 For two functions $u(x,t)$ and $v(x,t)$, convolution in both time and space is denoted by $**$ and defined as follows:
 \begin{align}(u**v)(x,t)=\int_{x'}\int_{t'}u(x,t)v(x-x',t-t')dx'dt'.\label{eqnconvts}
 \end{align} 
 The concentration of molecules of type $\mathtt{A}$ at location $x$ and time $t$ is denoted by $[\mathtt{A}](x,t)$.
 Similarly, we use $[\mathtt{B}](x,t)$ and $[\mathtt{C}](x,t)$ to denote the concentration of molecules of types $\mathtt{B}$ and $\mathtt{C}$ respectively. We write $[\mathtt{A}](x,t)$ as a series in the perturbation method and use 
 $[\mathtt{A}]_{i}(x,t)$ as the coefficient for the $i$th term in the expansion. 
 For simplicity, we sometimes show a function without its arguments, e.g. we write $[\mathtt{A}]_{i}$ or $[\mathtt{A}]$ instead of $[\mathtt{A}]_{i}(x,t)$ or $[\mathtt{A}](x,t)$. Finally, all chemical reactions are considered in the time interval $[0,T]$ for some $T>0$.

\section{Solving Reaction-Diffusion Equations via the Perturbation Method}\label{sec::generalperturbation}

The perturbation method provides an explicit analytical solution in the form of an infinite series of functions. One can approximate the solution by taking a finite number of terms from this series.  
 In this paper, we use the perturbation method to obtain an analytical model for the reaction-diffusion equations, which can be used for molecular communication. 
  
Our working example in the paper is as follows: consider two molecular transmitters $\mathsf \mathsf T_{\mathtt {A}}$ and $\mathsf \mathsf T_{\mathtt {B}}$, which release molecules $\mathtt{A}$ and $\mathtt{B}$, respectively. Molecules of type $\mathtt{A}$ and $\mathtt{B}$ react with each other in the medium and produce molecules of type $\mathtt{C}$ as follows:
\begin{align}\label{ex1:reaction}
\mathtt{A}  + \mathtt{B} 
\underset{\lambda_2}{\stackrel{\lambda_1}{\rightleftharpoons}} \mathtt{C}
\end{align}
We assume a molecular receiver $\mathsf \mathsf R_{\mathtt {C}}$, which measures the concentration of molecules of type $\mathtt{C}$ at its location in order to detect the messages of the two transmitters. 


For simplicity, assume that the transmitters and receiver are placed on a one-dimensional line, with $\mathsf T_{\mathtt {A}}$, $\mathsf T_{\mathtt {B}}$, and $\mathsf R_{\mathtt {C}}$ being located at $x=0$,  $x=d_\mathtt{B}$, and $x=d_R$ respectively. This assumption is just for the simplicity of exposition; the problem could be solved in two or three dimensions and we also provide simulation results in dimensions two and three in Section \ref{sec:Val}. We assume that the transmitters and receivers are small in size and view them as points on the real line; the transmitter and receiver are considered to be transparent, not impeding the diffusion of the released molecules (assumption of point sources is adopted for simplicity and is common in the MC literature, e.g., \cite{farahnak2018medium:j:16,noclosedform:2,farsad2016molecular:F:16}). This is depicted in Fig. \ref{fig:ex1}. Further details of the communication model are described in Section \ref{Sec:Mod} where we design a modulation scheme. In the rest of this section, we are merely interested in solving the reaction-diffusion equation describing reaction \eqref{ex1:reaction}.

\begin{figure}

	\centering
	\includegraphics[trim={1cm 0cm 0cm 0cm}, scale=0.25]{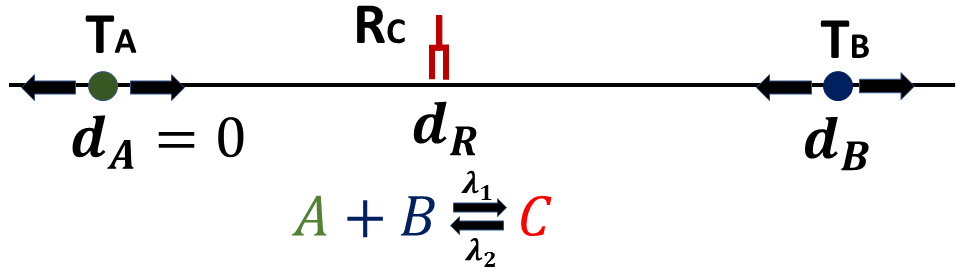}
	\caption{The system model}
	\label{fig:ex1}

	\vspace{-2.5em}

\end{figure}

Let $\gamma=\lambda_2/\lambda_1$ be the ratio of the  backward and forward  reaction rates, and set  $\lambda=\lambda_1,\lambda_2=\gamma\lambda$ for some $\lambda\geq 0$. 
The following equation describes the system dynamic:
\begin{align}
&\frac{\partial[\mathtt{A}]}{\partial t}=D_\mathtt{A}~\frac{\partial^{2}[\mathtt{A}]}{{\partial x}^{2}}-\lambda~[\mathtt{A}][\mathtt{B}]+\gamma \lambda  [\mathtt{C}]+f_\mathtt{A}(x,t),\label{A:ex1eq}\\
&	\frac{\partial[\mathtt{B}]}{\partial t}=D_\mathtt{B}~\frac{\partial^{2}[\mathtt{B}]}{{\partial x}^{2}}-\lambda~[\mathtt{A}][\mathtt{B}]+\gamma \lambda  [\mathtt{C}]+f_\mathtt{B}(x,t),\label{B:ex1eq}\\
&	\frac{\partial[\mathtt{C}]}{\partial t}=D_\mathtt{C}~\frac{\partial^{2}[\mathtt{C}]}{{\partial x}^{2}}+\lambda~[\mathtt{A}][\mathtt{B}]-\gamma \lambda  [\mathtt{C}]\label{C:ex1eq},
\end{align}
where $f_\mathtt{A}(x,t),f_\mathtt{B}(x,t)$ are the released concentration of molecules (called the transmission waveforms), and $D_\mathtt{A}, D_\mathtt{B}, D_\mathtt{C}$  are the diffusion constants. The initial and boundary conditions are set to zero.  We are interested in the solution for time $t\in[0,T]$.

For the special case of $\lambda=0$, \emph{i.e.,} when there is no reaction, the above system of equations is linear and tractable.
  Consider a solution of these equations in the form, 
  $[\mathtt{A}](x,t)=\sum_{i=0}^{\infty}\lambda^{i}
  [\mathtt{A}]_{i}(x,t)$, 
  $[\mathtt{B}](x,t)=\sum_{i=0}^{\infty}\lambda^{i}
  [\mathtt{B}]_{i}(x,t)$ and $[\mathtt{C}](x,t)=\sum_{i=0}^{\infty}\lambda^{i}
  [\mathtt{C}]_{i}(x,t)$ for some Taylor series coefficients 
  $[\mathtt{A}]_{i}(x,t)$, $[\mathtt{B}]_{i}(x,t)$ and $[\mathtt{C}]_{i}(x,t)$.
  By substituting these expressions in \eqref{A:ex1eq}, \eqref{B:ex1eq} and \eqref{C:ex1eq}, we obtain
  \begin{align}\label{New1h}
  &\frac{\partial}{\partial t}(\sum_{i=0}^{\infty}\lambda^{i}
  [\mathtt{A}]_{i})
  =D_\mathtt{A}~\frac{\partial^{2}}{{\partial x}^{2}}(\sum_{i=0}^{\infty}\lambda^{i}[\mathtt{A}]_{i})-\lambda (\sum_{i=0}^{\infty}\lambda^{i}
  [\mathtt{A}]_{i})(\sum_{i=0}^{\infty}\lambda^{i}[\mathtt{B}]_{i})+\gamma\lambda \sum_{i=0}^{\infty}\lambda^{i}[\mathtt{C}]_{i} +f_\mathtt{A}(x,t),
  \end{align}
  \begin{align}\label{New2h:b}
  &\frac{\partial}{\partial t}(\sum_{i=0}^{\infty}\lambda^{i}
  [\mathtt{B}]_{i})
  =D_\mathtt{B}~\frac{\partial^{2}}{{\partial x}^{2}}(\sum_{i=0}^{\infty}\lambda^{i}[\mathtt{B}]_{i})-\lambda (\sum_{i=0}^{\infty}\lambda^{i} 
  [\mathtt{A}]_{i})(\sum_{i=0}^{\infty}\lambda^{i}
  [\mathtt{B}]_{i})+\gamma\lambda \sum_{i=0}^{\infty}\lambda^{i}[\mathtt{C}]_{i}+
  f_\mathtt{B}(x,t),
  \end{align}
  \begin{align}\label{New3h}
  &\frac{\partial}{\partial t}(\sum_{i=0}^{\infty}\lambda^{i}
  [\mathtt{C}]_{i})
  =D_\mathtt{C}~\frac{\partial^{2}}{{\partial x}^{2}}(\sum_{i=0}^{\infty}\lambda^{i}[\mathtt{C}]_{i})+\lambda (\sum_{i=0}^{\infty}\lambda^{i} 
  [\mathtt{A}]_{i})(\sum_{i=0}^{\infty}\lambda^{i}
  [\mathtt{B}]_{i})-\gamma\lambda \sum_{i=0}^{\infty}\lambda^{i}[\mathtt{C}]_{i}
  .
  \end{align}
  
  Equations \eqref{New1h}, \eqref{New2h:b}, and \eqref{New3h} could be viewed as power series in $\lambda$ for fixed values of $x$ and $t$. Matching the coefficients of $\lambda^i$ on both sides of the equation, we obtain the following:
  \begin{align}\label{zerosol}
  \frac{\partial [\mathtt{A}]_{0}}{\partial t}=D_\mathtt{A}~\frac{\partial^{2}[\mathtt{A}]_{0}}{{\partial x}^{2}}+f_\mathtt{A}(x,t),~~
  \frac{\partial [\mathtt{B}]_{0}}{\partial t}=D_\mathtt{B}~\frac{\partial^{2} [\mathtt{B}]_{0}}{{\partial x}^{2}}+f_\mathtt{B}(x,t),~~
  \frac{\partial [\mathtt{C}]_{0}}{\partial t}=D_\mathtt{C}~\frac{\partial^{2} [\mathtt{C}]_{0}}{{\partial x}^{2}}.
  \end{align}
  For $i\geq 1$ we obtain
   \label{othersolA} 
  \begin{align}
  &\frac{\partial[\mathtt{A}]_{i}}{\partial t}=D_\mathtt{A}~\frac{\partial^{2}[\mathtt{A}]_{i}}{{\partial x}^{2}}-[\mathtt{A}]_{0:i-1}*^d
  [\mathtt{B}]_{0:i-1}+\gamma[\mathtt{C}]_{i-1},\label{othersolA} \\	
  &\frac{\partial[\mathtt{B}]_{i}}{\partial t}=D_\mathtt{B}~\frac{\partial^{2}[\mathtt{B}]_{i}}{{\partial x}^{2}}-[\mathtt{A}]_{0:i-1}*^d
  [\mathtt{B}]_{0:i-1}+\gamma[\mathtt{C}]_{i-1},\label{othersolB} \\
  &\frac{\partial[\mathtt{C}]_{i}}{\partial t}=D_\mathtt{C}~\frac{\partial^{2}[\mathtt{B}]_{i}}{{\partial x}^{2}}-[\mathtt{A}]_{0:i-1}*^d
  [\mathtt{B}]_{0:i-1}-\gamma[\mathtt{C}]_{i-1},\label{othersolC} 
  \end{align} 
  
  where we used the index convolution defined in \eqref{dicreteconv}.
  The functions $[\mathtt{A}]_{i}(x,t)$,  $[\mathtt{B}]_{i}(x,t)$ and $[\mathtt{C}]_{i}(x,t)$ could be find recursively by first computing $[\mathtt{A}]_{0}(x,t)$, $[\mathtt{B}]_{0}(x,t)$ and $[\mathtt{C}]_{0}(x,t)$  from \eqref{zerosol}, and then using \eqref{othersolA}-\eqref{othersolC} to compute $[\mathtt{A}]_{i}(x,t)$, $[\mathtt{B}]_{i}(x,t)$ and $[\mathtt{C}]_{i}(x,t)$  from $[\mathtt{A}]_{j}(x,t)$, $[\mathtt{B}]_{j}(x,t)$ and $[\mathtt{C}]_{j}(x,t)$ for $j\leq i-1$.
  
  \subsection{Limitation on the simulation time interval and a remedy}
The perturbation method provides recursive equations to find the functions $[\mathtt{A}]_{i}(x,t)$ and $[\mathtt{B}]_{i}(x,t)$.
However, one still needs to show that the power series $\sum_{i=0}^{\infty}\lambda^{i}
[\mathtt{A}]_{i}(x,t)$ and $\sum_{i=0}^{\infty}\lambda^{i}
[\mathtt{B}]_{i}(x,t)$ are convergent to functions that satisfy the original reaction-diffusion differential equation. Perturbation theory does not provide a general recipe for showing this convergence, and it should be done on a
 case-by-case basis. We show in Appendix \ref{AppA0 } that the series is convergent as long as $\lambda\leq \mathcal{O}(1/T)$, \emph{i.e.,} the radius of convergence of the series is proportional with the inverse of
 the total time period. In other words, we need $\lambda$ to be sufficiently small (or for $T$ to be sufficiently small). Provided that $T$ is sufficiently small, Appendix \ref{AppA0 } shows that the solution given by the perturbation method is equal to the true solution (satisfies the reaction-diffusion equations).

Next, we propose an approach to tackle the above shortcoming if we wish to compute the solution for larger values of $T$:  we can divide the time interval $T$ into $k$ subintervals of size $T/k$. Since the radius of convergence is of order $\mathcal{O}(\frac 1T)$, reducing the size of the interval by a factor of $k$ increases the radius of convergence by a multiplicative factor of $k$. Thus, we can first compute the solution in the interval $[0,T/k]$ and use the value at $T/k$ to compute the solution in the interval $[T/k,2T/k]$, etc.

\subsection{Solving equations \eqref{zerosol}-\eqref{othersolC}} 
  
  In order to solve \eqref{zerosol}, let  $\phi_\mathtt{A}(x,t)$ be the  solution  of the following equation with diffusion coefficient $D_\mathtt{A}$:
  \begin{align}\label{greenA}
  \frac{\partial \phi_\mathtt{A}}{\partial t}=D_\mathtt{A}~\frac{\partial^{2}\phi_\mathtt{A}}{{\partial x}^{2}}+\delta(x)\delta(t).
  \end{align}
  We have
  \begin{align}
  \phi_\mathtt{A}(x,t)=\frac{1}{\sqrt{4\pi D_\mathtt{A} t}}\exp(-\frac{x^2}{4 D_\mathtt{A} t}),~ x\in \mathbb{R},~t\geq 0.
  \end{align}
  Similarly, we define $\phi_\mathtt{B}(x,t)$ and  $\phi_\mathtt{C}(x,t)$  as the solutions of  
  \begin{align}\label{greenBC}
  \frac{\partial \phi_\mathtt{B}}{\partial t}=D_\mathtt{B}~\frac{\partial^{2}\phi_\mathtt{B}}{{\partial x}^{2}}+\delta(x)\delta(t), \qquad
  \frac{\partial \phi_\mathtt{C}}{\partial t}=D_\mathtt{C}~\frac{\partial^{2}\phi_\mathtt{C}}{{\partial x}^{2}}+\delta(x)\delta(t).
  \end{align}
  Then, 
  the solutions of  \eqref{zerosol}, 
  \eqref{othersolA}-\eqref{othersolC} are as follows:
  \begin{align}\label{ex1:A0,B0}
  &[\mathtt{A}]_{0}=\phi_\mathtt{A}**f_{\mathtt{A}},~~	[\mathtt{B}]_{0}=\phi_\mathtt{B}**f_{\mathtt{B}},~~[\mathtt{C}]_{0}=0,
  \end{align}
  and for $i\geq 1$, 
  \begin{align}
  &[\mathtt{A}]_{i}=-\phi_\mathtt{A}**([\mathtt{A}]_{0:i-1}*^d[\mathtt{B}]_{0:i-1}+\gamma[\mathtt{C}]_{i-1}),\label{recu-eq1}\\	
  &[\mathtt{B}]_{i}=-\phi_\mathtt{B}**([\mathtt{A}]_{0:i-1}*^d
  [\mathtt{B}]_{0:i-1}+\gamma[\mathtt{C}]_{i-1}),\label{recu-eq2}\\
  &[\mathtt{C}]_{i}=+\phi_\mathtt{C}**([\mathtt{A}]_{0:i-1}*^d
  [\mathtt{B}]_{0:i-1}-\gamma[\mathtt{C}]_{i-1}),\label{recu-eq3}
  \end{align}
  where  $**$ stands for the  two-dimensional convolution (see \eqref{eqnconvts}).
  
 In general, the perturbation method does not yield a closed-form analytical solution for the reaction-diffusion equation because it is not possible (in general) to find a closed-form solution for the recursive equations  \eqref{recu-eq1}-\eqref{recu-eq3}. However, the following example indicates a special instance where a closed-form solution can be found.
\begin{example}\label{exampleA1}

Consider the special choice of $f_{\mathtt{A}}(x,t)=f_{\mathtt{B}}(x,t)=1 ,x\in \mathbb{R}^{d}, t\geq 0  $, $D_{{\mathtt{A}}}=D_{{\mathtt{B}}}$, and $\gamma=0$. While these input signals are not interesting, the convolutions in  \eqref{recu-eq1}-\eqref{recu-eq3} can be explicitly worked out for this choice of input signals.

Using symmetry, we have $[\mathtt{A}]_i(x,t)=[\mathtt{B}]_i(x,t)$ for $i\geq 0$. The zero and first terms of perturbation  can be computed as follows:
\begin{align}
&[\mathtt{A}]_{0}(x,t)=\phi**f_{ \mathtt{A}}=\int_{0}^{t}\int_{x\in \mathbb{R}^{d}}\phi(x-x',t-t')f_{ \mathtt{A}}(x',t')dx'dt'=t. 
\end{align}
We prove, by induction, that
$[\mathtt{A}]_{i}(x,t)=(-1)^{i}\alpha_{i}t^{2i+1}$ where $\alpha_{i}=\sum_{j=0}^{i-1} \alpha_{j}\alpha_{i-1-j}/(2i+1), i\geq 1, $ and $\alpha_0=1$.  We have:
\begin{align}
[\mathtt{A}]_{i}&=-\phi**\sum_{j=0}^{i-1}[\mathtt{A}]_{j}[\mathtt{A}]_{i-1-j}\nonumber\\ &=-\int_{0}^{t}\int_{x\in \mathbb{R}^{d}}\phi(x-x',t-t')\sum_{j=0}^{i-1}(-1)^{j}\alpha_{j}t'^{2j+1}(-1)^{i-1-j}\alpha_{i-1-j}t'^{2i-2j-1} dx'dt'\nonumber\\&=(-1)^{i}\frac{\sum_{j=0}^{i-1} \alpha_{j}\alpha_{i-1-j}}{2i+1}t^{2i+1}. 
\end{align}
By induction, one can prove that $(1/3)^{i}\leq \alpha_{i}\leq (1/2)^i$. The solution of reaction-diffusion for this particular example is as follows: $[\mathtt{A}](x,t)=\sum_{i=0}^{\infty}\lambda^{i}[\mathtt{A}]_i=\sum_{i=0}^{\infty}\lambda^{i}(-1)^{i}\alpha_{i}t^{2i+1}$. 
\end{example}

\subsection{Interpretation in terms of Picard's iterative process}
An alternative way to understand the recursive equations \eqref{ex1:A0,B0}-\eqref{recu-eq3} is through 
Picard's iterative process. In this process, one converts the differential equation to an integral form and uses it to identify the Picard iteration operator. 
 Consider the reaction-diffusion equations in \eqref{A:ex1eq}-\eqref{C:ex1eq}. To convert this system of
 differential equations to an integral form, we utilize
 the fact that the solution of the differential equation $
  \frac{\partial[\mathtt{u}]}{\partial t}=D_\mathtt{u}~\frac{\partial^{2}[\mathtt{u}]}{{\partial x}^{2}}+g_{u}$
   with zero initial and boundary condition is $u=\phi_{u}**g_{u}$ where $\phi_{u}$ is the solution of the differential equation for an impulse input (impulse response). Thus, from \eqref{A:ex1eq}-\eqref{C:ex1eq} we obtain
\begin{align}
&[\mathtt{A}](x,t)=\phi_{\mathtt A}**(-\lambda[\mathtt{A}][\mathtt{B}]+\gamma\lambda[\mathtt{C}]+f_{\mathtt{A}}(x,t)),\\
&[\mathtt{B}](x,t)=\phi_{\mathtt B}**(-\lambda[\mathtt{A}][\mathtt{B}]+\gamma\lambda[\mathtt{C}]+f_{\mathtt{B}}(x,t)),\\
&[\mathtt{C}](x,t)=\phi_{\mathtt C}**(\lambda[\mathtt{A}][\mathtt{B}]-\gamma\lambda[\mathtt{C}]).
\end{align}
Therefore, the Picard iteration operator is obtained as follows:
\begin{align}
&\mathcal{\tau}:\mathcal{C}\times \mathcal{C}\times \mathcal{C} \mapsto \mathcal{C}\times \mathcal{C}\times \mathcal{C}\nonumber\\
&\mathcal{\tau}(\mathtt{a},\mathtt{b},\mathtt{c})=(\phi_\mathtt{A}**(-\lambda\mathtt{a}\mathtt{b}+\lambda\gamma\mathtt{c}+f_\mathtt{A}),\phi_\mathtt{B}**(-\lambda\mathtt{a}\mathtt{b}+\lambda\gamma\mathtt{c}+f_\mathtt{B}),\phi_\mathtt{C}**(\lambda\mathtt{a}\mathtt{b}-\gamma\lambda\mathtt{c})),
\end{align}
where $\mathcal{C}$ is the space of continuous functions. The solution of the reaction-diffusion system is a fixed point of the Picard iteration operator. Picard's iterative process suggests the following recursive algorithm to find a fixed point of the operator $\mathcal{\tau}$:
\begin{align}
(\mathtt{a}^{(n+1)},\mathtt{b}^{(n+1)},\mathtt{c}^{(n+1)})=\mathcal{\tau}(\mathtt{a}^{(n)},\mathtt{b}^{(n)},\mathtt{c}^{(n)}),~~~(\mathtt{a}^{(-1)},\mathtt{b}^{(-1)},\mathtt{c}^{(-1)})=(0,0,0),~~n\geq-1\label{banchiter},  
\end{align}
 where $(\mathtt{a}^{(n)},\mathtt{b}^{(n)},\mathtt{c}^{(n)})$ is the triple of functions that we obtain at the $n$th iteration. Picard's iterative process relates to the perturbation method as follows: using induction and comparing \eqref{recu-eq1}-\eqref{recu-eq3} with the recursive equations given in \eqref{banchiter}, one obtains
 \begin{align}
     \mathtt{a}^{(n)}&=\sum_{i=0}^{n}\lambda^{i}
 [\mathtt{A}]_{i}(x,t)&n=0,1\\
 \mathtt{a}^{(n)}&=\sum_{i=0}^{n}\lambda^{i}
 [\mathtt{A}]_{i}(x,t)+\mathcal{O}(\lambda^{n+1})&n\geq 2.
 \end{align}
 A similar statement holds for $\mathtt{b}^{(n)}$ and $\mathtt{c}^{(n)}$, showing that their first $n$ lower terms match the first $n$ lower terms of the Taylor series expansion in the perturbation method.
 
\subsection{Approximate Solutions}
 
In practice, we can approximate the true concentrations of molecules by computing  the first $n$ terms in the Taylor series expansion:
\begin{align}
[\mathtt{A}](x,t)\approx \sum_{i=0}^{n}\lambda^{i}
[\mathtt{A}]_{i}(x,t)
\end{align}
for some fixed $n$. One can stop at the $n$-term of the Taylor expansion when the contribution of the $(n+1)$th  term is negligible compared to the first $n$ terms. While $n$th term can be explicitly found in Example \ref{exampleA1}, for general reaction-diffusion equations, there is no (non-recursive) explicit expression for the terms in the Taylor series expansion. 
In Section \ref{sec:Val}, we illustrate the sufficiency of choosing $n=1$ or $n=2$ when the length of the time interval $T$ is  in the order of seconds. In the following, we provide a theoretical justification for choosing $n=1$ when the time interval $T$ is small: consider the reaction-diffusion equations
\begin{align}
&\frac{\partial[\mathtt{A}]}{\partial t}=D_\mathtt{A}~\frac{\partial^{2}[\mathtt{A}]}{{\partial x}^{2}}-\lambda~[\mathtt{A}][\mathtt{B}]+\gamma \lambda  [\mathtt{C}]+f_\mathtt{A}(x,t),\\
&	\frac{\partial[\mathtt{B}]}{\partial t}=D_\mathtt{B}~\frac{\partial^{2}[\mathtt{B}]}{{\partial x}^{2}}-\lambda~[\mathtt{A}][\mathtt{B}]+\gamma \lambda  [\mathtt{C}]+f_\mathtt{B}(x,t),\\
&	\frac{\partial[\mathtt{C}]}{\partial t}=D_\mathtt{C}~\frac{\partial^{2}[\mathtt{C}]}{{\partial x}^{2}}+\lambda~[\mathtt{A}][\mathtt{B}]-\gamma \lambda  [\mathtt{C}],
\end{align}
along with the initial conditions $
[\mathtt{A}](x,0)=g_{\mathtt{A}}(x), [\mathtt{B}](x,0)=g_{\mathtt{B}}(x)$ and $ [\mathtt{C}](x,0)=g_{\mathtt{C}}(x)$.
The  one-step finite difference iteration (in the finite difference method) for a small value of $t$ is as follows:
\begin{align}
[\mathtt{A}]_{FDM}(x,t)&=[\mathtt{A}](x,0)+t\frac{\partial[\mathtt{A}]}{\partial t}(x,0)
\nonumber\\&=g_{\mathtt{A}}(x)+t\left[  D_{\mathtt{A}}g_{\mathtt{A}}''(x)- \lambda g_{\mathtt{A}}(x)g_{\mathtt{B}}(x)+ \gamma\lambda g_{\mathtt{C}}(x)+  f_{\mathtt{A}}(x,0)\right]
\nonumber\\&=g_{\mathtt{A}}(x)+\mathcal{O}(t).\end{align}
Let $[\mathtt{A}]^{(1)}(x,t)=[\mathtt{A}]_{0}(x,t)+\lambda [\mathtt{A}]_{1}(x,t)$ be the approximate solution of the perturbation method up to the order one terms. In Appendix \ref{AppA:3} we show the following equivalence with the finite difference solution:
\begin{theorem}\label{thmN1}
The first-order approximation  $[\mathtt{A}]^{(1)}(x,t)=[\mathtt{A}]_{0}(x,t)+\lambda [\mathtt{A}]_{1}(x,t)$ and the finite difference method are equivalent for small values of $t$, \emph{i.e.,} $  [\mathtt{A}]^{(1)}(x,t)= [\mathtt{A}]_{FDM}(x,t)+\mathcal{O}(t^2)$.
   A similar statement holds for molecules of type $\mathtt{B}$ and $\mathtt{C}$. 
\end{theorem}
As illustrated in Section
    \ref{sec:Val}, the time steps do not necessarily need to be very small for the perturbation solution to closely track the true solution (and can be in the order of seconds).  Therefore, the perturbation method can be understood as a generalization of the finite difference method that allows for longer time steps.

  \section{Modulation Design}\label{Sec:Mod}
  In Section \ref{sec::generalperturbation}, a closed-form solution for the reaction-diffusion equation \eqref{ex1:reaction} was obtained. In this section, we use this solution to find the optimal release pattern of molecules by the transmitters in order to minimize the error probability of the communication system.

	Suppose transmitter $ \mathsf T_{\mathtt {A}}$ intends to send a message bit $M_\mathtt{A}\in\{0,1\}$ to the receiver. The transmitter $ \mathsf T_{\mathtt {A}}$ releases molecules of type $[\mathtt {A}]$ into the medium according to $f_\mathtt{A}^{0}(x,t)=a_0(t)\delta(x), 0\leq t\leq T$ if $M_\mathtt{A}=0$, and according to $f_\mathtt{A}^{1}(x,t)=a_1(t)\delta(x), 0\leq t\leq T$ if 
	$M_\mathtt{A}=1$. 
In other words, the waveform $a_{0}(t)$ is the released concentration of molecules of type $\mathtt A$ when $M_\mathtt{A}=0$, and $a_{1}(t)$
 is the released density when $M_\mathtt{A}=1$. Similarly, suppose that transmitter $ \mathsf T_{\mathtt {B}}$
 encodes message $M_\mathtt{B}\in\{0,1\}$ by $f_{\mathtt{B}}^{j}(x,t)=b_{j}(t)
 \delta(x-d_{\mathtt{B}}), 0\leq t\leq T, j=0,1$, where the waveforms $b_{0}(t)$ and $b_{1}(t)$ are defined similarly for transmitter $\mathsf T_{\mathtt {B}}$. The total amount of released molecules of types $\mathtt A$ and $\mathtt B$ during the transmission period $T$ is assumed to be at most $s_{\mathtt A}$ and $s_{\mathtt B}$ respectively,  \emph{i.e.,}
 \begin{align}\int_{0}^T a_i(t)dt\leq s_{\mathtt A},~\int_{0}^T b_i(t)dt\leq s_{\mathtt B}, \quad i=0,1.\label{eqn:New6}
\end{align}
  
  The input messages
 ${M}_\mathtt{A}, {M}_\mathtt{B}$ are assumed to be uniform Bernoulli random variables.  
 Receiver $\mathsf R_{\mathtt {C}}$
samples the density of molecules of type $\mathtt{C}$ at its location at the end of the time slot at time $t=T$. The receiver is assumed to be transparent and suffers from the particle counting noise. More specifically, the receiver is assumed to get a sample from $\mathsf{Poisson}\big(V\cdot [\mathtt{C}](d_{R},T)\big)$ where $[\mathtt{C}](d_{R},T)$ is the density of molecules of type $\mathtt C$ at time $T$ and location $d_{R}$, and $V$ is a constant.

 Using this observation, $\mathsf R_{\mathtt {C}}$ 
outputs its estimate of the transmitted message pair $(\hat{M}_\mathtt{A},\hat{M}_\mathtt{B})\in \{00,01,10,11\}$. The probability of error is 
 \begin{equation}
 Pr(e)=Pr\{ (M_{\mathtt{A}}, M_{\mathtt{B}})\neq (\hat{M}_\mathtt{A},\hat{M}_\mathtt{B})\}.\label{eqnErrorProb1}
 \end{equation}
 Having transmitted messages 
 $a_{m_{\mathtt{A}}}(t), b_{m_{\mathtt{B}}}(t)$
 	 for a message pair $(m_{\mathtt{A}}, m_{\mathtt{B}})$,  let $\rho_{m_{\mathtt{A}}, m_{\mathtt{B}}}$ denote the density of molecules of type $\mathtt{C}$ at the time of sampling. The receiver aims to recover $\hat m_{\mathtt{A}}, \hat m_{\mathtt{B}}$ using a sample $\mathsf{Poisson}(V\cdot \rho_{m_{\mathtt{A}}, m_{\mathtt{B}}})$.  Our goal is to design nonnegative signals $a_{i}(t), b_{j}(t)$ 	satisfying \eqref{eqn:New6}
	to minimize the error probability. The key utility of the technique given in  Section \ref{sec::generalperturbation} is that it provides a closed-form expression for 
	$\rho_{m_{\mathtt{A}}, m_{\mathtt{B}}}$ in terms of $a_i(t)$ and $b_j(t)$. 

	An order $n$ approximation to the reaction-diffusion is when we consider the reaction-diffusion equations up to the terms of order $\lambda^n$ in the Taylor expansion. We say that waveforms $a_i(t)$ and $b_j(t)$ are 
 order $n$ optimal waveforms if they minimize the error probability for the approximate reaction-diffusion equations up to order $n$. For low reaction rates, we are mainly interested in order one optimal waveforms. In the following theorem, we show that for low reaction rates, one possible optimal strategy is for 
	$ \mathsf T_{\mathtt {A}},\mathsf T_{\mathtt {B}}$  to release molecules in at most two time instances in a bursty fashion. The location of these instantaneous releases is determined by the message it wants to transmit. In other words, waveform $a_{i}(t)$ is the sum of two Dirac's delta functions. 
\begin{theorem}\label{th::1}
Consider approximating the solution of the reaction-diffusion equations with the perturbation solution up to the 
first-order term. Consider the problem of choosing waveforms $a_i(t), b_i(t)$ to minimize the error probability. Then, restricting  minimization to waveforms $a_i(t), b_i(t)$ with the following structure does not change the minimum value of error probability:
	\begin{align}
	&a_0(t)=\hat{a}_{01}\delta(t-t_1^{[a_0]})+\hat{a}_{02}\delta(t-t_2^{[a_0]}),~
	 a_1(t)=\hat{a}_{11}\delta(t-t_1^{[
		a_1]})+\hat{a}_{12}\delta(t-t_2^{[a_1]}),\\
	& b_0(t)=\hat{b}_{01}\delta(t-t_1^
	{[b_0]})+\hat{b}_{02}\delta(t-t_2^
	{[b_0]}),~~
	b_1(t)=\hat{b}_{11}\delta(t-t_1^{[b_1]})+\hat{b}_{12}\delta(t-t_2^{[b_1]}),
	\end{align}
	 for some non-negative constants  $\hat{a}_{ij},\hat{b}_{ij}$ satisfying
	 $\sum_{j}\hat{a}_{ij}\leq s_{\mathtt A}$ and
	 $\sum_{j}\hat{b}_{ij}\leq s_{\mathtt B}$ for $i=0,1$,
	  and for some
	$t_j^{[a_i]},t_j^{[b_i]}\in[0,T]$. 
\end{theorem}
\begin{proof}
By substituting $f_\mathtt{A}^{i}(x,t),f_\mathtt{B}^{i}(x,t), i=0,1$ in \eqref{ex1:A0,B0} and \eqref{recu-eq1}-\eqref{recu-eq3}, one obtains an explicit formula for (the first-order approximation of) the concentration of molecules $\mathtt{C}$ in terms of $f_\mathtt{A}^{i}(x,t),f_\mathtt{B}^{i}(x,t), i=0,1$. In particular, the concentration $[\mathtt{C}](d_{R},T)$ at receiver's location at the sampling time $T$ is as follows:
\begin{align}
[\mathtt{C}](d_
{R},T)=\lambda\int_{0}^{T}
\int_{-\infty}^{+\infty}\int_{0}^{t'}
\int_{0}^{t'}&
\phi_{\mathtt{C}}(d_R-x',T-t')
\phi_{\mathtt{A}}(x',t'-t_1)
\phi_{\mathtt{B}}(x'-d_{\mathtt{B}},t'-t_2)\nonumber\\&
a_{M_{\mathtt{A}}}(t_1)
b_{M_{\mathtt{B}}}(t_2)
dt_{1}dt_{2}dx'dt'.
\end{align}
Observe that 
$[\mathtt{C}](d_
{R},T)$ has a \emph{bilinear} form with respect to $a_{M_{\mathtt{A}}}(t)$ and $b_{M_{\mathtt{B}}}(t)$ . In the other words, if $a_{M_{\mathtt{A}}}(t)$  is kept fixed then $[\mathtt{C}](d_
{R},T)$ is linear with respect to $b_{M_{\mathtt{B}}}(t)$ and vice versa.

Let
\begin{equation}\label{densities::ex1:final}
\rho=\left\{ {\begin{array}{*{20}{llll}}
	\rho_{00}=[\mathtt{C}](d_{R},T)& & \text{if}&(M_{\mathtt{A}},M_{\mathtt{B}})=(0,0)\\
	\rho_{01}=[\mathtt{C}](d_{R},T)& & \text{if}&(M_{\mathtt{A}},M_{\mathtt{B}})=(0,1)\\
	\rho_{10}=[\mathtt{C}](d_{R},T)& & \text{if}&(M_{\mathtt{A}},M_{\mathtt{B}})=(1,0)\\	\rho_{11}=[\mathtt{C}](d_{R},T)& & \text{if}&(M_{\mathtt{A}},M_{\mathtt{B}})=(1,1)
	\\
	\end{array}}\right.
\end{equation}
Using \eqref{densities::ex1:final} we can compute $
(\rho_{00},\rho_{01},\rho_{10},\rho_{11})$ for any given $a_{0}(t), a_{1}(t), b_0(t)$ and $b_1(t)$. The receiver aims to recover $\hat m_{\mathtt{A}}, \hat m_{\mathtt{B}}$ using a sample $\mathsf{Poisson}(V\rho_{m_{\mathtt{A}}, m_{\mathtt{B}}})$. Minimizing the error probability (given in \eqref{eqnErrorProb1})  is equivalent to solving a Poisson hypothesis testing problem (see Appendix \ref{AppB} for a review). Our goal is to design nonnegative signals $a_{i}(t), b_{j}(t)$ 	satisfying \eqref{eqn:New6}
to minimize the error probability of this Poisson hypothesis testing problem. 

 Note that the tuple $
(\rho_{00},\rho_{01},\rho_{10},\rho_{11})$ forms a \emph{sufficient statistic} for computing the error probability. In other words, if we change the signals $a_{i}(t), b_{j}(t)$ in such a way that the four numbers $
(\rho_{00},\rho_{01},\rho_{10},\rho_{11})$ are unchanged, the error probability remains unchanged. 
Take  signals $a_i(t)$ and $b_j(t)$ for $i,j\in\{0,1\}$ that minimize the error probability. We claim that we can find $\hat a_0(t)$ of the form 
$\hat a_0(t)=\hat{a}_{01}\delta(t-t_1^{[a_0]})+\hat{a}_{02}\delta(t-t_2^{[a_0]})$
for some $\hat{a}_{01}$ and $\hat{a}_{02}$
such that if we replace the $(a_0(t), a_1(t), b_0(t), b_1(t))$ by $(\hat a_0(t), a_1(t), b_0(t), b_1(t)),$ the corresponding values of  $
(\rho_{00},\rho_{01},\rho_{10},\rho_{11})$ remain unchanged. Moreover, $\hat a_0(t)$ satisfies the power constraint \eqref{eqn:New6}. Since the error probability depends on the waveforms only through the values of
$
(\rho_{00},\rho_{01},\rho_{10},\rho_{11})$, we deduce that replacing $a_0(t)$ with $\hat a_0(t)$ does not change the error probability. Since signals $a_i(t)$ and $b_j(t)$ for $i,j\in\{0,1\}$ minimized the error probability, we deduce that using $\hat a_0(t)$ would also achieve the minimum possible probability of error. 
By a similar argument we can reduce the support of waveforms $a_1(t)$, $b_0(t)$, and $b_1(t)$ one by one while fixing the other waveforms. This completes the proof.

It remains to show that changing $a_0(t)$ by  $\hat{a}_0(t)$ while preserving the values of $
(\rho_{00},\rho_{01},\rho_{10},\rho_{11})$ is possible. 
Replacing $a_0(t)$ by  $\hat{a}_0(t)$ while fixing $a_1(t)$, $b_0(t)$, and $b_1(t)$ may only change $\rho_{00}$ and $\rho_{01}$. The values of $\rho_{10}$ and $\rho_{11}$ are only functions of $a_1(t)$, $b_0(t)$, and $b_1(t)$ which we are fixing. Moreover, $\rho_{00}$ and $\rho_{01}$ are linear functions of $a_0(t)$ when we fix $b_0(t)$ and $b_1(t)$. 
In other words, one can find functions $f_1(t)$ and $f_2(t)$ such that
\begin{equation}
    \rho_{00}=\int_{t\in[0,T]}a_0(t)f_1(t)dt,\qquad \rho_{01}=\int_{t\in[0,T]}a_0(t)f_2(t)dt.\label{eqnDeff1}
    \end{equation}

The power constraint \eqref{eqn:New6} is also a linear constraint on $a_0(t)$. We would like to change $a_0(t)$ to $\hat a_0(t)$ in such a way that two linear constraints corresponding to $\rho_{00}$ and $\rho_{01}$ are preserved, and the power of $\hat a_0(t)$ is less than or equal to the power of  $a_0(t)$. Existence of $\hat{a}_0(t)$ with these  properties follows from Lemma \ref{SupportLemma} given in Appendix \ref{AppC}  (with the choice of $f(t)=1$, $n=2$, and $f_1(t)$ and $f_2(t)$ given in \eqref{eqnDeff1}), once we view $a_0(t)$ as an unnormalized probability distribution.

\end{proof}
In order to determine parameters in the statement of Theorem \ref{th::1} we need to optimize over constants $\hat{a}_{ij}$, $\hat{b}_{ij}$, $t_j^{[a_i]},t_j^{[b_i]}$.
 From the first-order equation, 
 $\rho_{ij}$ (the concentration when the first transmitter sends bit $i\in\{0,1\}$ and the second transmitter sends the bit $j\in\{0,1\}$) equals
\begin{equation}
 \begin{array}{*{20}{l}} 
	\rho_{00}=\sum_{i,j=1}^{2}\hat{a}_{0i}\hat{b}_{0j}G(t_{i}^{[a_0]},t_{j}^{[b_0]})
	\\
	\rho_{01}=\sum_{i,j=1}^{2}\hat{a}_{0i}\hat{b}_{1j}G(t_{i}^{[a_0]},t_{j}^{[b_1]})\\
	\rho_{10}=\sum_{i,j=1}^{2}\hat{a}_{1i}\hat{b}_{0j}G(t_{i}^{[a_1]},t_{j}^{[b_0]})\\
	\rho_{11}=\sum_{i,j=1}^{2}\hat{a}_{1i}\hat{b}_{1j}G(t_{i}^{[a_1]},t_{j}^{[b_1]}),
	\end{array}
\end{equation}
where
\begin{equation}
G(t_i,t_j)=\lambda\phi_{\mathtt{C}}(x,t)**\big(\phi_{\mathtt{A}}(x,t-t_i) 
\phi_{\mathtt{B}}(x-d_{\mathtt{B}},t-t_j)\big)
|_{t=t_s,x=d_R}.
\end{equation}
To recover the messages of the  transmitters, the receiver needs to solve a Poisson hypothesis testing problem (see Appendix \ref{AppB}) with  four  hypotheses: $\rho=V\rho_{00}$ or $\rho=V\rho_{01}$ or $\rho=V\rho_{10}$ or $\rho=V\rho_{11}$ for a constant $V$. The error probability of this hypothesis testing problem should be minimized with respect to $\hat{a}_{ij}$, $\hat{b}_{ij}$, $t_j^{[a_i]},t_j^{[b_i]}$.

As an example, consider  the following parameters in Table\ref{tabj}:
\begin{table}[H]
  	\centering
  	\caption{Parameters}
	 \label{tabj}
\begin{tabular}{ |c|c|c| } 

	\hline
	\multirow{2}{*}{$(D_{\mathtt{A}},D_{\mathtt{B}},D_{\mathtt{C}})[m^{2}s^{-1}]$}   &\multirow{2}{*} {$10^{-9}\times(1,1,1)$}\\
	&\\
	\hline 
	\multirow{2}{*} {$\lambda[molecules^{-1}.m^{3}.s^{-1}]$}   & \multirow{2}{*}{$ 10^{-23}$}\\
	&\\
	\hline
	\multirow{2}{*}{$V[m^{3}]$}   &\multirow{2}{*} {$10^{-11}$}\\
	&\\ 
	\hline
	\multirow{2}{*}{$T[s]$ }& \multirow{2}{*}{$3$}  \\ 
	&\\
	\hline
	
	\multirow{2}{*}{$d_{\mathtt{B}} [m]$} & \multirow{2}{*}{$5\times 10^{-5}\times(2,0,0) $ } \\
	&\\ 
	\hline
	\multirow{2}{*}{$d_{R}[m]$} & \multirow{2}{*}{$5\times 10^{-5}\times(1,0,0)$} \\
	&\\
		\hline
	\multirow{2}{*}{$s_{\mathtt{A}}[molecules.m^{-3}],s_{\mathtt{B}}[molecules.m^{-3}])$} & \multirow{2}{*}{$10^7\times(1,1)$} \\
	&\\
	\hline 
\end{tabular}
\end{table}
Simulation results yield the optimal values for $\hat{a}_{ij}$, $\hat{b}_{ij}$, $t_j^{[a_i]},t_j^{[b_i]}$ as in Table\ref{tab78}
 \begin{table}[H]
  	\centering
  	\caption{Values of unknown parameters.}
  		\label{tab78}
	\begin{tabular}{ |c|c|c| }  
		\hline
		\multirow{2}{*}{$(\hat{a}_{01},\hat{a}_{02})$}   &\multirow{2}{*} {$(1.13\times 10^6,8.82\times 10^6)$}\\
		&\\
		\hline 
		\multirow{2}{*} {$(\hat{a}_{11},\hat{a}_{12})$}   & \multirow{2}{*}{$ (2.17\times 10^6,1.07\times 10^6)$}\\
		&\\
		\hline
		\multirow{2}{*}{$(\hat{b}_{01},\hat{b}_{02})$}   &\multirow{2}{*} {$(0.97\times 10^6,3.2\times 10^6)$}\\
		&\\ 
		\hline
		\multirow{2}{*}{$(\hat{b}_{11},\hat{b}_{12})$ }& \multirow{2}{*}{$(3.52\times 10^6,6.46\times 10^6)$}  \\ 
		&\\
		\hline
		\multirow{2}{*}{$(t_{1}^{[a_0]},t_{2}^{[a_0]})$} & \multirow{2}{*}{$(0,2) $ } \\
		&\\
		\hline
		\multirow{2}{*}{$(t_{1}^{[a_1]},t_{2}^{[a_1]})$} & \multirow{2}{*}{$(0,1) $ } \\
		&\\ 
		\hline
		\multirow{2}{*}{$(t_{1}^{[b_0]},t_{2}^{[b_0]})$} & \multirow{2}{*}{$(1,2)$} \\
		&\\ 
		\hline
		\multirow{2}{*}{$t_{1}^{[b_1]},t_{2}^{[b_1]}$} & \multirow{2}{*}{$(1,2)$} \\
		&\\ 
		\hline
		\multirow{2}{*}{$p_{e}$} & \multirow{2}{*}{$2.6\times 10^{-5}$} \\
		&\\
		\hline 
	\end{tabular}
\end{table}


 Figs. \ref{probabilty1} and  \ref{probabilty2} show the optimal error probability (as defined in \eqref{eqnErrorProb1}) as a function of $s=s_{\mathtt{A}}=s_{\mathtt{B}}$ for the above parameters.

 Fig. \ref{pulsevsdelta} compares  the error probabilities of optimal waveforms (sum of two delta functions) with a pulse release pattern, \emph{i.e.,} when we consider four waveforms $(f_{\mathtt{A}}^{0}(t)=a_{0}\mathbbm{1}_{0\leq t\leq T},f_{\mathtt{A}}^{1}(t)=a_{1}\mathbbm{1}_{0\leq t\leq T},f_{\mathtt{B}}^{0}(t)=b_{0}\mathbbm{1}_{0\leq t\leq T},f_{\mathtt{B}}^{1}(t)=b_{1}\mathbbm{1}_{0\leq t\leq T})$ as input signals to reaction diffusion equations. The error probability of the pulse waveforms is obtained by minimizing the error probability over pulse amplitudes $a_0, a_1\leq s_{\mathtt{A}}/T$ and $b_0, b_1\leq s_{\mathtt{B}}/T$.

\begin{figure}[H]
	\centering
		\includegraphics[scale=.75]{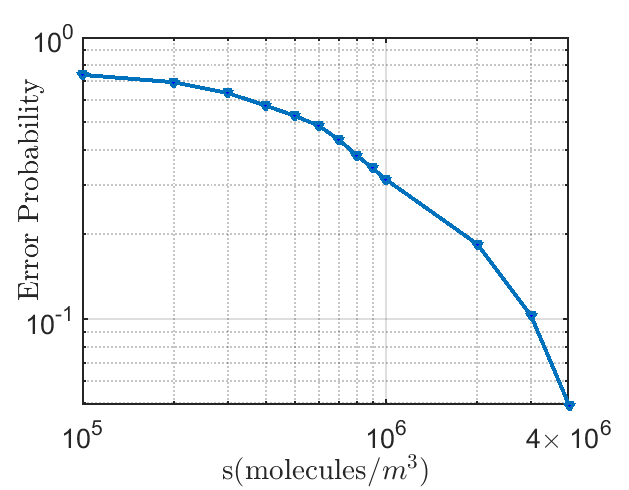}
		\caption{Error probability of the optimal waveforms for
		$10^5\leq s \leq 4\times 10^6$ [molecules.$m^{-3}$].}
		\label{probabilty1}
 		\vspace{-3em}
	\end{figure}
	\begin{figure}[H]
		\centering
	\includegraphics[scale=.75]{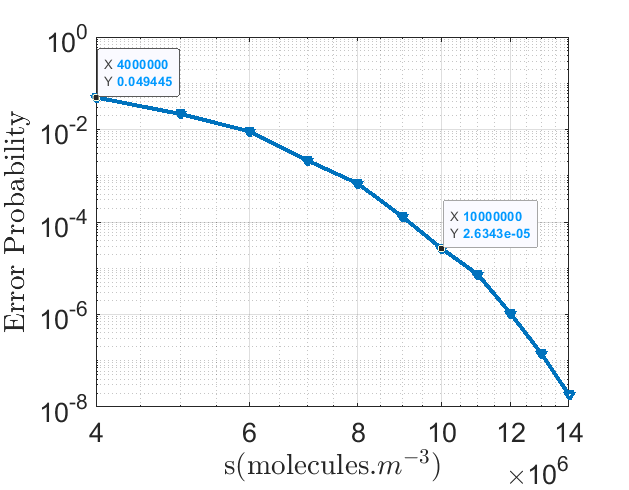}
		\caption{Error Probability of the optimal waveforms
		for
		$4\times 10^6\leq s \leq 1.4\times 10^7$
[molecules.$m^{-3}$].}
		\label{probabilty2}
\end{figure}

\begin{figure}[H]
	\centering

		\includegraphics[trim={1cm 0cm 0cm 0cm}, scale=.75]{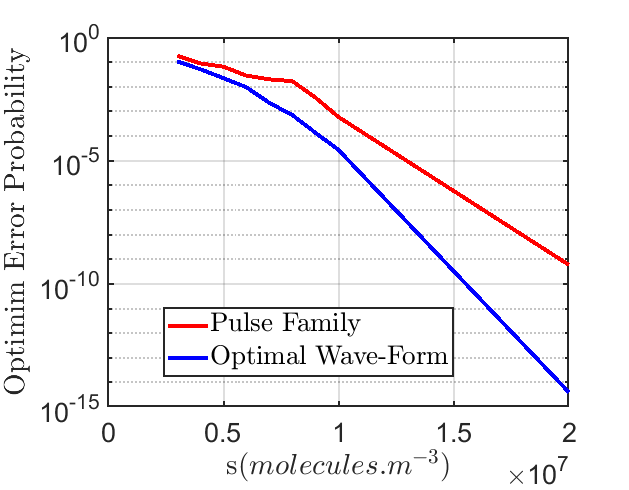}
		\caption{Comparison between optimal wave-form and pulse family.}
		\label{pulsevsdelta}
	
		\vspace{-3em}
	\end{figure}
	
	\begin{figure}[H]
 		\centering
 	\includegraphics[trim={1cm 0cm 0cm 0cm}, scale=.75]{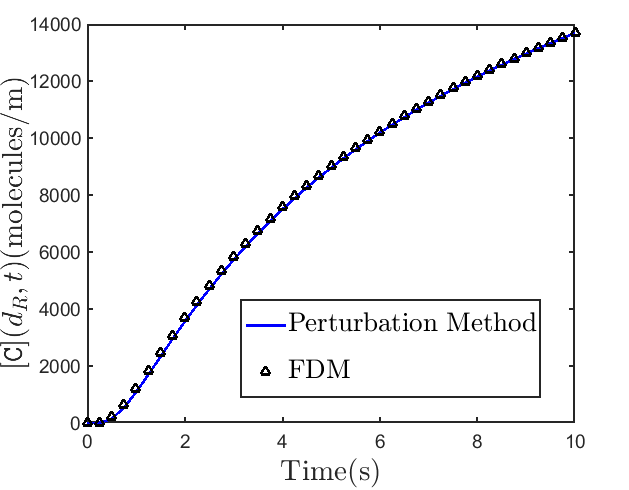}
 	\caption{First-order approximation for the first set parameters }
 	\label{fig:1e-22,unequal,Nx=64,first}
\end{figure}

 \section{Simulation}
 \label{sec:Val}
 In this section, we validate our method through numerical simulation. We need to specify the diffusion coefficients, reaction rate $\lambda$, input signals $f_{\mathtt A}$ and $f_{\mathtt B}$ as well as the locations of the transmitters. Previous works in the literature commonly take the number of released molecules to be either in the order of $10^{-14}$ moles (which is $10^9$ molecules) \cite{arjmandi2019diffusive,farahnak2018medium:j:16,cao2019chemical:j:ro}, or in the order of $10^{-3}$  moles \cite{farsad2016molecular:F:16}. The reaction rates in these works are generally in the order of $1$ to $10^8$ $\text{mole}^{-1}.m^{3}.s^{-1}$ or equivalently $10^{-23}$ to $10^{-15}$ $\text{molecules}^{-1}.m^{3}.s^{-1}$ \cite{cao2019chemical:j:ro,farsad2016molecular:F:16,bi2019chemical:j:9,chang2005physical}.
 
 We run three simulations for three choices of parameters for the reaction given in \eqref{ex1:reaction}. The exact solutions are obtained by the finite difference method (FDM). The first simulation borrows its system parameters from \cite{cao2019chemical:j:ro} as in Table \ref{tab1}:
 \begin{table}
	\centering
	\caption{First set of parameters.}
	\label{tab1}
	\begin{tabular}{ |c|c|c| }  
		\hline
		\multirow{2}{*}{$(D_{\mathtt{A}},D_{\mathtt{B}},D_{\mathtt{C}})[m^{2}s^{-1}]$}   &\multirow{2}{*} {$10^{-10}\times(10,7,1)$}\\
		&\\
		\hline 
		\multirow{2}{*} {$\lambda[molecules^{-1}.m.s^{-1}]$}   & \multirow{2}{*}{$ 10^{-22}$}\\
			&\\
		\hline 
		\multirow{2}{*} {$\gamma[s^{-1}]$}   & \multirow{2}{*}{$0$}\\
		&\\
		\hline
		\multirow{2}{*}{$f_{\mathtt{A}}(x,t)$}   &\multirow{2}{*} {$5\times 10^{8}\delta(x)\delta(t)$}\\
		&\\ 
		\hline
		\multirow{2}{*}{$f_{\mathtt{B}}(x,t)$ }& \multirow{2}{*}{$2.4\times 10^{9}\delta(x-d_{\mathtt{B}})\delta(t)$}  \\ 
		&\\
		\hline
		\multirow{2}{*}{$d_{\mathtt{A}} [m]$} & \multirow{2}{*}{$0 $ } \\
		&\\
		\hline
		\multirow{2}{*}{$d_{\mathtt{B}} [m]$} & \multirow{2}{*}{$ 10^{-4} $ } \\
		&\\ 
		\hline
		\multirow{2}{*}{$d_{R}[m]$} & \multirow{2}{*}{$5\times 10^{-5}$} \\
		&\\
		\hline 
	\end{tabular}
\end{table}

In the second simulation we use the same parameters as in the first simulation, but make the three diffusion constants be the same $(D_{\mathtt{A}}=D_{\mathtt{B}}=D_{\mathtt{C}}=10^{-9}[m^2/s])$.
 Figs. \ref{fig:1e-22,unequal,Nx=64,first} and \ref{fig:1e-22,equal,Nx=120,first} are plotted based on the first and second set of system parameters respectively. They depict the concentration of the molecule of type $\mathtt{C}$ at the receiver's location calculated by the first-order term of perturbation method along with the true solution (computed via FDM).
 It can be seen that for these system parameters, the first-order term of the perturbation series is closely following the true solution.

 \begin{figure}[H] 
 	
 	\centering
 \includegraphics[trim={1cm 0cm 0cm 0cm}, scale=.75]{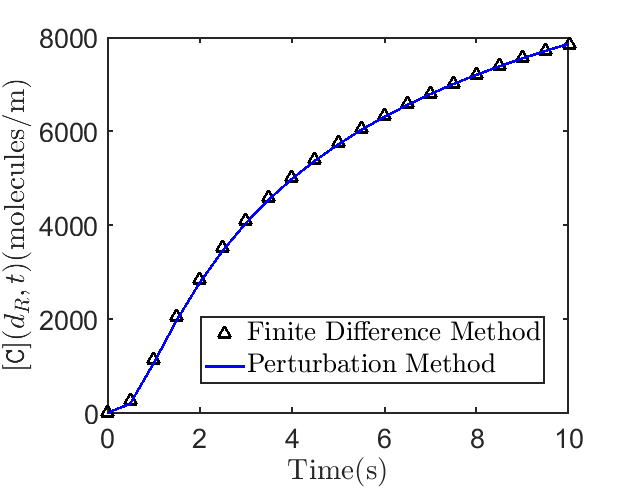}
 \caption{First-order approximation for the second set parameters}
 \label{fig:1e-22,equal,Nx=120,first}
 \vspace{-3em}
 \end{figure}

 \begin{figure}[H]
		\centering
		\includegraphics[trim={1cm 0cm 0cm 0cm}, scale=.75]{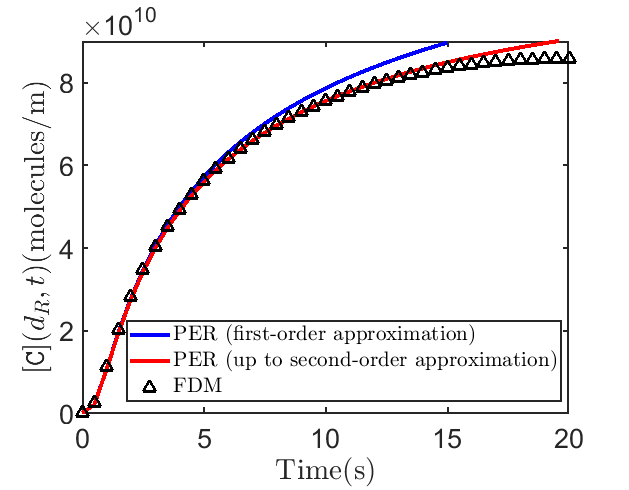}
		\caption{First-order and second-order approximate solutions for $\lambda=10^{-15}$.}
		\label{fig:1e-15,equal,Nx=120,first,T=15}
 \end{figure}

  Using the same parameters as in the second simulation, one can observe that the first-order term of perturbation series is still accurate as long as $\lambda\leq 10^{-16} [\text{molecules}^{-1}.m.s^{-1}]$.  Fig. \ref{fig:1e-15,equal,Nx=120,first,T=15} plots the curves for $\lambda=10^{-15}$. As one can observe while the first-order term of  the perturbation series does not track the exact solution for $t>5s$, while the second-order term tracks the true solution in $t<15s$.
 For $t>15s$, the perturbation solution up to the second-order term does not track the exact solution well, hence to obtain an accurate solution we have to consider third or higher-order terms in the perturbation solution.  This observation is consistent with the theoretical analysis given in Appendix \ref{AppA0 }, where the convergence rate of perturbation solution depends on the length of the time interval $T$. We also remark that in many prior works on chemical reactions in MC, the observation time $T$ is  rather small:   \cite{arjmandi2019diffusive} uses $T\leq 0.1 s$,  \cite{bi2019chemical:j:9} uses $T\leq 4 s$;  \cite{cao2019chemical:j:ro} uses $T\leq 10 s$, while in other cases  $T$ is in taken of order $\mu s$ \cite{farahnak2018medium:j:16,mosayebi2017type:F:21}.

Fig. \ref{relative1} fixes $T=10[s]$ and plots the relative error of the first-order approximation solution of $[\mathtt{C}]$ for the second set of parameters in terms of $\lambda$.  Fig. \ref{relative21} draws a similar curve for the second-order approximation of the perturbation solution. 
These figures confirm (the theoretical result) that for sufficiently small reaction rates,  the relative error is negligible.  For instance, to have a relative error of at most $0.01$, the first-order solution can be used for
 $\lambda<3\times10^{-16}$ (Fig. \ref{relative1}) while the second-order solution is valid until $\lambda<3\times10^{-15}$ (Fig. \ref{relative21}).

Next, let us fix the relative error to be $0.05$ and define the permissible time interval as 
\begin{align}
T_{\max}=\max\left\{T:\left|\frac{[\mathtt{C}]_{FDM}(t)-[\mathtt{C}]_{PER-Order-1}(t)}{[\mathtt{C}]_{FDM}(t)}\right|\leq .05, \qquad \forall t\in [0,T]\right\}.\label{TmaxDef}\end{align}
In other words, $T_{\max}$ is the maximum simulation time interval for which the relative error is less than 5 percent.
 Fig. \ref{maxt2} shows the $T_{\max}$ for the first-order approximation of $[\mathtt{C}]$ in terms of $\lambda$.

\begin{figure}[H]
		\centering
		\includegraphics[trim={1cm 0cm 0cm 0cm}, scale=0.75]{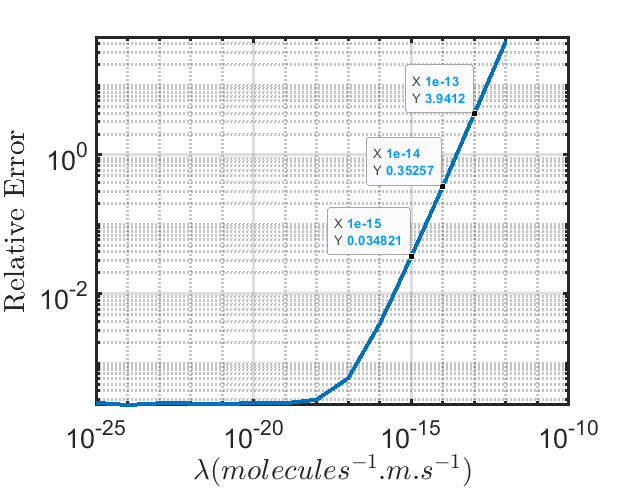}
		\caption{Relative error for  the first-order approximation.}
		\label{relative1}
		\vspace{-3em}
	\end{figure}
	\begin{figure}[H]
		\centering
		\includegraphics[trim={1cm 0cm 0cm 0cm}, scale=0.75]{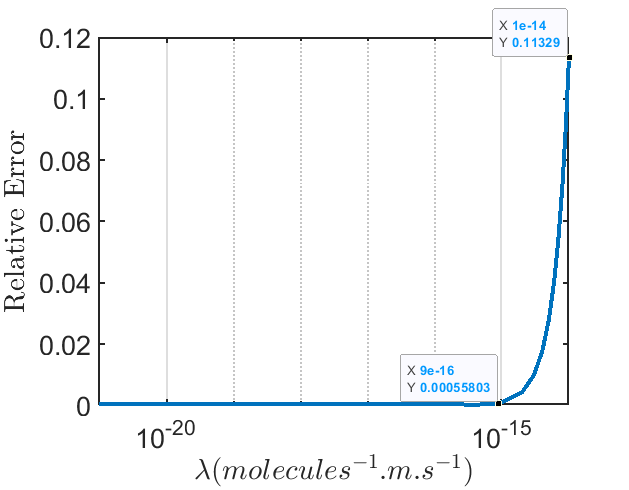}
		\caption{Relative error for up to second-order approximation.}
		\label{relative21}
\end{figure}

	\begin{figure}[H]
		\centering
		\includegraphics[trim={1cm 0cm 0cm 0cm}, scale=0.75]{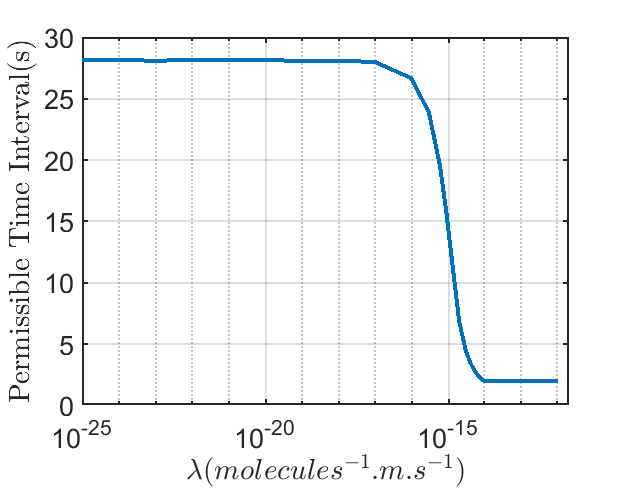}
		\caption{Permissible simulation time interval in terms of $\lambda$.}
		\label{maxt2}
		\vspace{-3em}
	\end{figure}
	\begin{figure}[H]
		\centering
		\includegraphics[trim={1cm 0cm 0cm 0cm}, scale=0.75]{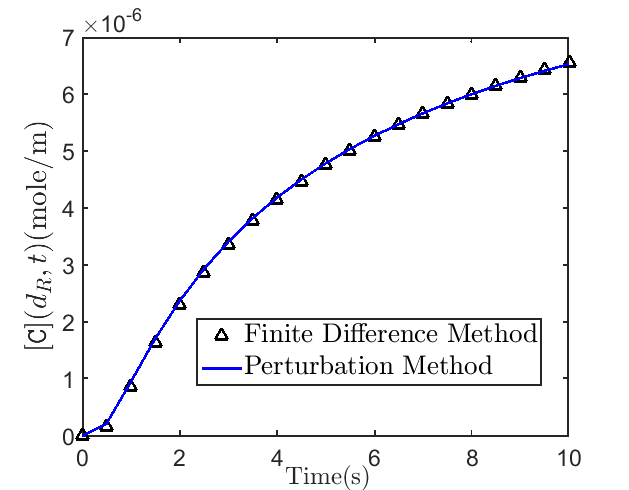}
		\caption{First-order approximation curve for the third set parameters.}
		\label{fig:.1,equal,Na=1micromole,first}
\end{figure}
In some applications, a significantly larger number of molecules (in the order of moles) are released into the medium. In the third simulation, we consider this case and take the following parameters in Table \ref{tab2}: 
\begin{table}
	\centering
	\caption{Third set of parameters}
	\label{tab2}
	\begin{tabular}{ |c|c|c| }  
		\hline
		\multirow{2}{*}{$(D_{\mathtt{A}},D_{\mathtt{B}},D_{\mathtt{C}})[m^{2}.s^{-1}]$}   &\multirow{2}{*} {$10^{-9}\times(1,1,1)$}\\
		&\\
		\hline 
		\multirow{2}{*} {$\lambda[molecules^{-1}.m.s^{-1}]$}   & \multirow{2}{*}{$ 10^{-24}/6.02214$}\\
		&\\
		\hline 
		\multirow{2}{*} {$\gamma[s^{-1}]$}   & \multirow{2}{*}{$0$}\\
		&\\
		\hline
		\multirow{2}{*}{$f_{\mathtt{A}}(x,t)$}   &\multirow{2}{*} {$6.02214\times 10^{17}\delta(x)\delta(t)$}\\
		&\\ 
		\hline
		\multirow{2}{*}{$f_{\mathtt{B}}(x,t)$ }& \multirow{2}{*}{$6.02214\times 10^{17}\delta(x-d_{\mathtt{B}})\delta(t)$}  \\ 
		&\\
		\hline
		\multirow{2}{*}{$d_{\mathtt{A}} [m]$} & \multirow{2}{*}{$0 $ } \\
		&\\
		\hline
		\multirow{2}{*}{$d_{\mathtt{B}} [m]$} & \multirow{2}{*}{$ 10^{-4} $ } \\
		&\\ 
		\hline
		\multirow{2}{*}{$d_{R}[m]$} & \multirow{2}{*}{$5\times 10^{-5}$} \\
		&\\
		\hline 
	\end{tabular}
\end{table}

 Fig. \ref{fig:.1,equal,Na=1micromole,first} shows that the first term of the perturbation series is close to the exact solution.

 
	


\textbf{The effect of dimension:} 
While we assumed the transmitter and receivers to lie on a one-dimensional line, the problem can be solved in higher dimensions in a similar manner. In Figs. \ref{dim2} and \ref{dim3}, we compare the perturbation method (up to the first-order approximation) and the true solution in two and three dimensions. The parameters are the same as the second set of parameters with $\lambda[{molecules}^{-1}.m^{3}.s^{-1}]=10^{-23}$.  As we see our solution is still accurate. We observe that as the dimension of the medium increases, concentrations decay faster. This is due to the extra degrees of freedom in the dispersion of molecules.  
\begin{figure}[H]
		\centering
		\includegraphics[trim={1cm 0cm 0cm 0cm}, scale=.75]{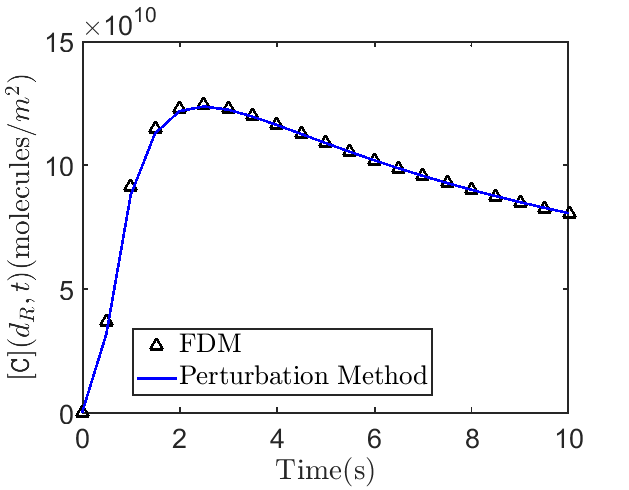}
		\caption{2d medium, first-order approximation}
		\label{dim2}
		\vspace{-3em}
	\end{figure}
	\begin{figure}[H]
		\centering
		\includegraphics[trim={1cm 0cm 0cm 0cm}, scale=0.75]{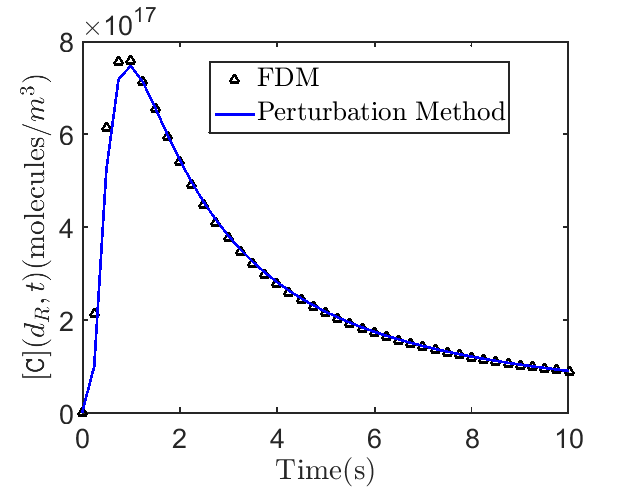}
		\caption{3d medium, first-order approximation}
		\label{dim3}
\end{figure}

\textbf{The effect of diffusion coefficients:}
In  Fig. \ref{DconC} we plot  $[\mathtt{C}]$ for different values of $D_\mathtt{C}$, for  the second set of parameters. Observe that $[\mathtt{C}]$ is a decreasing function of $D_\mathtt{C}$. Fig. \ref{maxt} plots $T_{\max}$ (as defined in \eqref{TmaxDef}) versus  $D_\mathtt{C}$ for the same parameters, showing that increasing  $D_\mathtt{C}$ decreases $T_{\max}$.

\begin{figure}[H]
		\centering
		\includegraphics[trim={1cm 0cm 0cm 0cm}, scale=0.75]{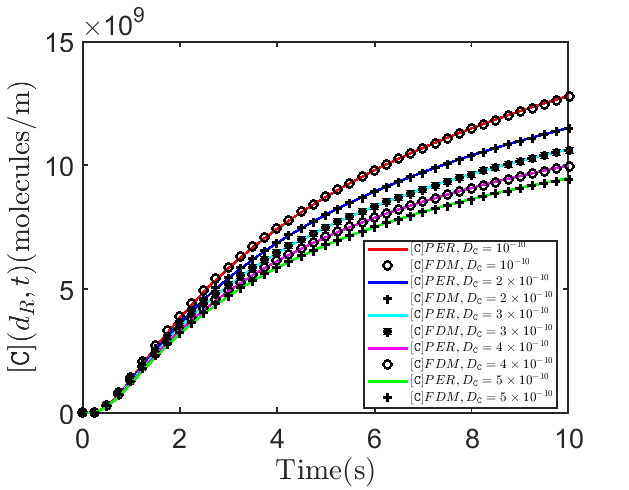}
		\caption{Effect of $D_\mathtt{C}$ on $[\mathtt{C}]$ }
		\label{DconC}
		\vspace{-3em}
	\end{figure}
	\begin{figure}[H]
		\centering
		\includegraphics[trim={1cm 0cm 0cm 0cm}, scale=0.75]{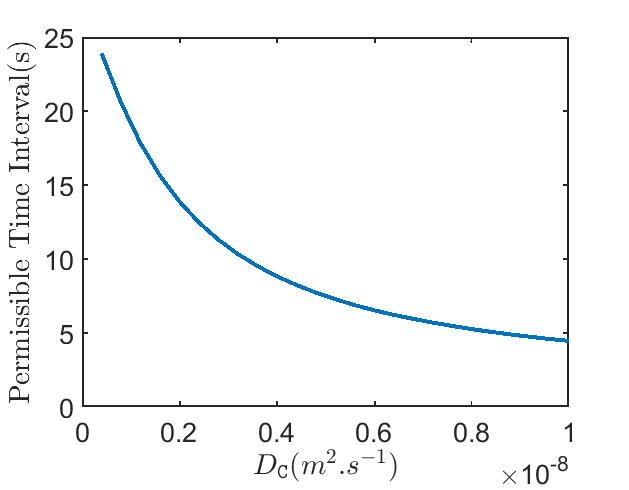}
		\caption{Impact of diffusion coefficient on the permissible simulation time interval}
		\label{maxt}
\end{figure}

\section{Conclusion and Future Work}\label{Conclusion and Future Work}
We addressed the difficulty of designing modulation due to the lack of existence of closed-form solutions for the reaction-diffusion equation by providing an approximate solution to the reaction-diffusion equations. We observed that for many choices of system parameters our solution is accurate. Also, 
 we observed that the accuracy of our solution depends on time observation  $T$. Using the proposed model, 
we designed optimal waveforms for a multiple-access setting. In the Appendix\ref{ex2,3} of this paper two more examples are considered (Example II and Example III).
 In Example II,  we consider a communication system with one transmitter and one receiver. The transmitter can send molecules of type $\mathtt{A}$ or $\mathtt{B}$, while the receiver can only measure molecules of type $\mathtt{A}$. The  medium has  molecules of a different type $\mathtt{C}$ that can react with both $\mathtt{A}$ and $\mathtt{B}$. The transmitter releases molecules of type $\mathtt{B}$ in order to ``clean up" the medium of molecules of type $\mathtt{C}$ so that molecules of type $\mathtt{A}$ can reach the receiver without being dissolved through reaction with $\mathtt{C}$ as they travel from the transmitter to the receiver. In Example III, we consider a two-way communication system with two transceiver nodes. In other words, each node aims to both send and receiver information from the other node. Transmitters use different molecule types and these molecules can react with each other. The chemical reaction in this scenario is destructive because it reduces the concentration of molecules at the two transceivers.

Many other settings could be studied and left for future work. Firstly, we only considered transparent receivers in this work. The literature on MC also considers absorbing or ligand/reactive receivers. Since the perturbation method is broadly applicable to non-linear differential equations,  the proposed framework can be also applied to absorbing or reactive receivers with more work. An absorbing receiver adds a boundary condition to the reaction-diffusion equations (the concentration of absorbed molecules being zero at the receiver's location). On the other hand, the reactive receiver adds an ordinary differential equation to the system's equations, for receptors on the surface of the receiver. 
 Further, one could study the design of optimal waveforms when taking multiple samples in the transmission time slot (instead of just one sample, as in this paper) or study optimal waveform design for channels with ISI. These studies are possible and left as future work. Finally, there are other existing approaches in the literature for obtaining  semianalytic solutions of chemical reaction-diffusion equations. Studying these approaches which might have a higher radius of convergence is also left as future work.

\appendices
\section{Convergence Analysis}\label{AppA0 }

To show that the power series $\sum_{i=0}^{\infty}\lambda^{i}
[\mathtt{A}]_{i}(x,t)$, $\sum_{i=0}^{\infty}\lambda^{i}
[\mathtt{B}]_{i}(x,t)$ and $\sum_{i=0}^{\infty}\lambda^{i}
[\mathtt{C}]_{i}(x,t)$ are convergent to functions that satisfy the original reaction-diffusion differential equation, 
we consider time $t\in[0,T]$ for some fixed $T$ and prove that 
\begin{align}\label{eqnNewN14}
&\sum_{i=0}^{\infty}\lambda^{i}
[\mathtt{A}]_{i}(x,t),\quad \sum_{i=0}^{\infty}\lambda^{i}
\frac{\partial}{\partial x} [\mathtt{A}]_{i}(x,t),\quad
\sum_{i=0}^{\infty}\lambda^{i}
\frac{\partial^2}{\partial x^2}[\mathtt{A}]_{i}(x,t),\quad\sum_{i=0}^{\infty}\lambda^{i}
\frac{\partial}{\partial t}[\mathtt{A}]_{i}(x,t)
\end{align}
uniformly converge over all $\lambda, x, t$. Similarly, we prove that the corresponding power series for molecules of types $\mathtt B$ and $\mathtt C$ also uniformly converge.

Uniform convergence of the series given in equations \eqref{eqnNewN14}
follows from Lemma \ref{LemmaBounds} given below. For a function $f_{\mathtt{A}}(x,t)$ defined for $x\in\mathbb{R}$ and $t\in[0,T]$, we define
\begin{align}&\lVert f_{\mathtt{A}}\lVert_{\infty}=\sup_{x,t\in[0,T]}|f_{\mathtt{A}}(x,t)|,
\lVert f_{\mathtt{A}}(x,0)\lVert_{\infty}=\sup_{x}|f_{\mathtt{A}}(x,0)|,
\lVert f_{\mathtt{A}}(x,t)\lVert_{1}=\int_{0}^{T}\int_{x}|f_{\mathtt{A}}(x,t)| dx dt.
\end{align}
\begin{lem}\label{LemmaBounds}
Let
\begin{align}
&M_0=\max\{T\lVert f_{\mathtt{A}}\lVert_{\infty},T\lVert f_{\mathtt{B}}\lVert_{\infty}\},
N_0=\max\big\{T\lVert \frac{\partial f_{\mathtt{A}}}{\partial t}\lVert_{\infty}+\lVert f_{\mathtt{A}}(x,0)\lVert_{\infty},
 T\lVert \frac{\partial f_{\mathtt{B}}}{\partial t}\lVert_{\infty}+\lVert f_{\mathtt{B}}(x,0)\lVert_{\infty}\big\},\nonumber\\
&G_0=\max\{\sqrt{\frac{4T}{\pi D_\mathtt{A}}}\lVert f_{\mathtt{A}}\lVert_{\infty},\sqrt{\frac{4T}{\pi D_\mathtt{B}}}\lVert f_{\mathtt{B}}\lVert_{\infty}\},~
H_0=\max\{\sqrt{\frac{4T}{\pi D_\mathtt{A}}}\lVert \frac{\partial f_{\mathtt{A}}}{\partial x}\lVert_{\infty},\sqrt{\frac{4T}{\pi D_\mathtt{B}}}\lVert \frac{\partial f_{\mathtt{B}}}{\partial x}\lVert_{\infty}\},\nonumber\\
&\sigma=\max\{\rVert\frac{\partial\phi_\mathtt{A}}{\partial x}\lVert_{1},\rVert\frac{\partial\phi_\mathtt{B}}{\partial x}\lVert_{1},\rVert\frac{\partial\phi_\mathtt{C}}{\partial x}\lVert_{1}\}=\sqrt{\frac{4T}{\pi\min\{D_\mathtt{A},D_\mathtt{B},D_{\mathtt{C}}\}}}.
\end{align} 
Then, for any $0<\lambda<\frac{1}{T(12M_0+10\gamma)}$, we have the following equations for any  $i\geq 1$
\begin{align}
&\lambda^{i}\rVert[\mathtt{A}]_{i}\lVert_{\infty}\leq \frac{M_0}{2^{i}},~
\lambda^{i}\rVert\frac{\partial [\mathtt{A}]_{i}}{\partial t}\lVert_{\infty}\leq \frac{N_0}{3(2^{i})},~\lambda^{i}\rVert\frac{\partial [\mathtt{A}]_{i}}{\partial x}\lVert_{\infty}\leq \frac{\sigma M_0 i}{4T(2^i)},\\&\lambda^{i}\rVert\frac{\partial^2 [\mathtt{A}]_{i}}{\partial x^2}\lVert_{\infty}\leq(\frac{1}{2})^{i}
(\frac{\sigma(2M_0+\gamma) G_0}{4T(M_0+\gamma)}+\frac{\sigma^2(2M_{0}^{2}+
	\gamma M_{0})}{32T^2(M_{0}+\gamma)}i(i-1)).
\end{align}
\end{lem}
Using this lemma and assuming $|\lambda|<\frac{1}{T(12M_0+10\gamma)}$ where $M_0$ is defined in  Lemma \ref{LemmaBounds}, we obtain
$\sum_{i=N}^{\infty}
\lambda^{i}\rVert[\mathtt{A}]_{i}\lVert_{\infty}\leq M_0/2^{N-1}$.
Since $M_0/(2^{N-1})$ is a universal upper bound that does not depend on $(\lambda, x, t)$,  the power series $\sum_{i=0}^{\infty}\lambda^{i}
[\mathtt{A}]_{i}(x,t)$ will uniformly converge. Proof of the convergence of the other power series given above is similar.	
It remains to prove Lemma \ref{LemmaBounds}. 
	
\begin{proof}[Proof of Lemma \ref{LemmaBounds}]
	The proof is by induction on $i$.
	Suppose $M_i,i\geq1$, is an upper bound on  $\rVert[\mathtt{A}]_i\lVert_{\infty}$, $\lVert[\mathtt{B}]_i\rVert_{\infty}$, and $\lVert[\mathtt{C}]_i\rVert_{\infty}$   we have:
	\begin{align}
	\rVert[\mathtt{A}]_i\lVert_{\infty}&=\rVert \phi_\mathtt{A}**([\mathtt{A}]_{0:i-1}*^d[\mathtt{B}]_{0:i-1}-\gamma[\mathtt{C}]_{i-1})\lVert_{\infty}\leq\rVert\phi_{\mathtt{A}}\lVert_{1}(\sum_{j=0}^{i-1}\rVert [\mathtt{A}]_{j}\lVert_{\infty}
	\rVert [\mathtt{B}]_{i-1-j}\lVert_{\infty}+\gamma\rVert [\mathtt{C}]_{i-1}\lVert_{\infty})\nonumber\\
	&\leq
	T (\sum_{j=0}^{i-1} M_jM_{i-1-j}+\gamma M_{i-1})\leq T(1+\frac{\gamma}{M_{0}})\sum_{j=0}^{i-1} M_jM_{i-1-j}
	\end{align}
	A similar bound for $\rVert[\mathtt{B}]_i\lVert_{\infty}$ and $\rVert[\mathtt{C}]_i\lVert_{\infty}$ can be written. Thus, using induction we can set $M_i=T(1+\frac{\gamma}{M_{0}}) \sum_{j=0}^{i-1} M_jM_{i-1-j}$. The solution of this recursive equation is given in the following form:
$	M_i=(T+\frac{\gamma T}{M_0})^iM_{0}^{i+1}\mathcal{C}_i
	,~i\geq 0,
$	where $\mathcal{C}_i$ is the  Catalan number and has an explicit formula:
$	\mathcal{C}_i=\frac{(2i)!}{(i+1)!(i)!}.
$	Using the Sterling formula, we have $
	\mathcal{C}_i\leq 4^i.
	$
	Hence for $i\geq 0$ we obtain:
	\begin{align}
	\lambda^{i}\max(\rVert[\mathtt{A}]_i\lVert_{\infty}, \rVert[\mathtt{B}]_i\lVert_{\infty})&\leq M_0(\frac{4M_0+4\gamma}{12M_0+10\gamma})^i\leq \frac{M_0}{2^i}.
	\end{align}

Suppose $N_i$ for $i\geq1$ is an upper bound on $\rVert\frac{\partial[\mathtt{A}]_i}{\partial t}\lVert_{\infty},\lVert
\frac{\partial[\mathtt{B}]_i}{\partial t}\rVert_{\infty}$ and $\lVert
\frac{\partial[\mathtt{C}]_i}{\partial t}\rVert_{\infty}$ .
We have:
\begin{align}
\rVert\frac{\partial[\mathtt{A}]_i}{\partial t}\lVert_{\infty}&=\rVert \phi_\mathtt{A}**\frac{\partial}{\partial t}([\mathtt{A}]_{0:i-1}*^d[\mathtt{B}]_{0:i-1}-\gamma [\mathtt{C}]_{i-1})\lVert_{\infty}\nonumber\\ &\leq \rVert\phi_{\mathtt{A}}\lVert_{1}(2\sum_{j=0}^{i-1} \rVert\frac{\partial[\mathtt{A}]_j}{\partial t}\lVert_{\infty}
\rVert [\mathtt{B}]_{i-1-j}\lVert_{\infty}+\gamma \rVert\frac{\partial[\mathtt{C}]_{i-1}}{\partial t}\lVert_{\infty})\nonumber\\&\leq (2T+\frac{\gamma T}{M_0})\sum_{j=0}^{i-1}N_{j}M_{i-1-j}\leq
(2TM_0+\gamma T) \sum_{j=0}^{i-1} N_j(4TM_0)^{i-1-j}.
\end{align}
A similar equation can be written for $\rVert\frac{\partial[\mathtt{B}]_i}{\partial t}\lVert_{\infty}$ and $\rVert\frac{\partial[\mathtt{C}]_i}{\partial t}\lVert_{\infty}$. 
Thus, setting
$N_i=(2TM_0+\gamma T) \sum_{j=0}^{i-1} N_j(4TM_0)^{i-1-j}$ yields a valid upper bound on 
$\rVert\frac{\partial[\mathtt{A}]_i}{\partial t}\lVert_{\infty},\lVert
\frac{\partial[\mathtt{B}]_i}{\partial t}\rVert_{\infty}$ and $\lVert
\frac{\partial[\mathtt{C}]_i}{\partial t}\rVert_{\infty}$ 
by induction.
The solution of this recursive equation is as follows:
\begin{equation}
N_i=N_0(\frac{2M_0+\gamma}{6M_0+5\gamma})(6TM_0+5T\gamma)^{i},~i\geq 1.
\end{equation}
For $0<\lambda<\frac{1}{T(12M_0+10\gamma)}$, we obtain that
\begin{equation}
\lambda^{i}\max(\rVert\frac{\partial [\mathtt{A}]_{i}}{\partial t}\lVert_{\infty},\rVert\frac{\partial [\mathtt{B}]_{i}}{\partial t}\lVert_{\infty}\big)\leq \frac{N_0}{3(2^{i})}.
\end{equation}

Next, we have
\begin{align}
\rVert\frac{\partial[\mathtt{A}]_i}{\partial x}\lVert_{\infty}&=\rVert \frac{\partial\phi_\mathtt{A}}{\partial x} **([\mathtt{A}]_{0:i-1}*^d[\mathtt{B}]_{0:i-1}-\gamma [\mathtt{C}]_{i-1})\lVert_{\infty}\nonumber\\&\leq \rVert\frac{\partial\phi_\mathtt{A}}{\partial x}\lVert_{1}(\sum_{j=0}^{i-1} \rVert[\mathtt{A}]_j\lVert_{\infty}
\rVert [\mathtt{B}]_{i-1-j}\lVert_{\infty}+\gamma \rVert [\mathtt{C}]_{i-1}\lVert_{\infty} )\leq \sigma(1+\frac{\gamma}{M_0})\sum_{j=0}^{i-1}M_{j}M_{i-1-j}\nonumber\\
&\leq
\frac{\sigma M_0}{4T} \sum_{j=0}^{i-1} (4T(M_0+\gamma))^{i}
\leq \frac{\sigma M_0}{4T}i
(4T(M_0+\gamma))^{i},
\end{align}
thus, 
$
G_i=\frac{\sigma M_0}{4T}i
(4T(M_0+\gamma))^{i},
$
for $i\geq 1$, serves as an upper bound on 
$\rVert\frac{\partial[\mathtt{A}]_i}{\partial x}\lVert_{\infty}$ (and by a similar argument $\lVert
\frac{\partial[\mathtt{B}]_i}{\partial x}\rVert_{\infty}$ and $\lVert
\frac{\partial[\mathtt{C}]_i}{\partial x}\rVert_{\infty}$ ).
Finally, 
\begin{align}
\rVert\frac{\partial^2[\mathtt{A}]_i}{\partial x^2}\lVert_{\infty}&=\rVert
\frac{\partial \phi_\mathtt{A}}{\partial x}
**(\frac{\partial}{\partial x}([\mathtt{A}]_{0:i-1}*^d[\mathtt{B}]_{0:i-1}-\gamma [\mathtt{C}]_{i-1})\lVert_{\infty}
\nonumber\\&\leq \rVert\frac{\partial \phi_\mathtt{A}}{\partial x} \lVert_{1}(2\sum_{j=0}^{i-1} \rVert\frac{\partial[\mathtt{A}]_j}{\partial x}\lVert_{\infty}
\rVert [\mathtt{B}]_{i-1-j}\lVert_{\infty}+
\gamma \rVert\frac{\partial[\mathtt{C}]_{i-1}}{\partial x}\lVert_{\infty}
)\leq\sigma(2\sum_{j=0}^{i-1}
G_jM_{i-1-j}+\gamma G_{i-1})
\nonumber\\&\leq\sigma(2+\frac{\gamma}{M_0}) \sum_{j=0}^{i-1}
G_jM_{i-1-j}=\sigma(2+\frac{\gamma}{M_0})(G_0 M_{i-1}+\sum_{j=1}^{i-1} G_jM_{i-1-j})\nonumber\\&\leq(\frac{\sigma(2M_0+\gamma) G_0}{4T(M_0+\gamma)}+\frac{\sigma^2(2M_{0}^{2}+
	\gamma M_{0})}{32T^2(M_{0}+\gamma)}i(i-1))\times
(4T(M_0+\gamma))^{i}.
\end{align}
Therefore, 
$
H_i=(\frac{\sigma(2M_0+\gamma) G_0}{4T(M_0+\gamma)}+\frac{\sigma^2(2M_{0}^{2}+
	\gamma M_{0})}{32T^2(M_{0}+\gamma)}i(i-1))\times
(4T(M_0+\gamma))^{i}, i\geq 1,
$
is an upper bound on $\rVert\frac{\partial^2[\mathtt{A}]_i}{\partial x^2}\lVert_{\infty}$ (an on $\lVert
\frac{\partial^2[\mathtt{B}]_i}{\partial x^2}\rVert_{\infty}$ and $\lVert
\frac{\partial^2[\mathtt{C}]_i}{\partial x^2}\rVert_{\infty}$ by a similar argument). The proof is complete.
\end{proof}

\section{Proof of Theorem \ref{thmN1}  }\label{AppA:3}

The  zero-order terms in the perturbation method satisfy
\begin{align}
&\frac{\partial [\mathtt{A}]_{0}}{\partial t}=D_\mathtt{A}~\frac{\partial^{2}[\mathtt{A}]_{0}}{{\partial x}^{2}}+f_\mathtt{A}(x,t),~~
\frac{\partial [\mathtt{B}]_{0}}{\partial t}=D_\mathtt{B}~\frac{\partial^{2} [\mathtt{B}]_{0}}{{\partial x}^{2}}+f_\mathtt{B}(x,t),~~
\frac{\partial [\mathtt{C}]_{0}}{\partial t}=D_\mathtt{C}~\frac{\partial^{2} [\mathtt{C}]_{0}}{{\partial x}^{2}},\\
&[\mathtt{A}]_{0}(x,0)=g_{\mathtt{A}}(x), ~[\mathtt{B}]_{0}(x,0)=g_{\mathtt{B}}(x), ~[\mathtt{C}]_{0}(x,0)=g_{\mathtt{C}}(x).
\end{align}
Setting $t=0$ in above equations we have:
\begin{align}
&\frac{\partial [\mathtt{A}]_{0}}{\partial t}(x,0)=D_\mathtt{A}~\frac{\partial^{2}[\mathtt{A}]_{0}(x,0)}{{\partial x}^{2}}+f_\mathtt{A}(x,0)=D_\mathtt{A}
g_{\mathtt{A}}''(x)+f_\mathtt{A}(x,0).
\end{align}
Thus,
$[\mathtt{A}]_{0}(x,t )=[\mathtt{A}]_{0}(x,0)+\frac{\partial [\mathtt{A}]_{0}}{\partial t}(x,0)t +o(t )=g_{\mathtt{A}}(x)+t  D_\mathtt{A}
g_{\mathtt{A}}''(x)+t  f_\mathtt{A}(x,0)+o(t ).
$ For other terms in perturbation solution we have ($i\geq 1$):
\begin{align}
&\frac{\partial[\mathtt{A}]_{i}}{\partial t}=D_\mathtt{A}~\frac{\partial^{2}[\mathtt{A}]_{i}}{{\partial x}^{2}}-[\mathtt{A}]_{0:i-1}*^d
[\mathtt{B}]_{0:i-1}+\gamma[\mathtt{C}]_{i-1},~\quad[\mathtt{A}]_{i}(x,0)=0,\\	
&\frac{\partial[\mathtt{B}]_{i}}{\partial t}=D_\mathtt{B}~\frac{\partial^{2}[\mathtt{B}]_{i}}{{\partial x}^{2}}-[\mathtt{A}]_{0:i-1}*^d
[\mathtt{B}]_{0:i-1}+\gamma[\mathtt{C}]_{i-1},~\quad[\mathtt{B}]_{i}(x,0)=0,\\
&\frac{\partial[\mathtt{C}]_{i}}{\partial t}=D_\mathtt{C}~\frac{\partial^{2}[\mathtt{B}]_{i}}{{\partial x}^{2}}-[\mathtt{A}]_{0:i-1}*^d
[\mathtt{B}]_{0:i-1}-\gamma[\mathtt{C}]_{i-1},~\quad[\mathtt{C}]_{i}(x,0)=0.
\end{align}
Setting $t=0$ in above equations we obtain:
\begin{align}
\frac{\partial[\mathtt{A}]_{1}}{\partial t}(x,0)=
-g_{\mathtt{A}}(x)g_{\mathtt{B}}(x)+\gamma g_{\mathtt{C}}(x),\qquad
\frac{\partial[\mathtt{A}]_{i}}{\partial t}(x,0)=0,\quad~~i\geq 2.
\end{align} 
From the Taylor series expansion we have:
\begin{align}
&[\mathtt{A}]_{1}(x,t )=[\mathtt{A}]_{1}(x,0)+\frac{\partial [\mathtt{A}]_{1}}{\partial t}(x,0)t +o(t )=-t  g_{\mathtt{A}}(x)g_{\mathtt{B}}(x)+t  \gamma g_{\mathtt{C}}(x)+o(t ),\\
&[\mathtt{A}]_{i}(x,t )=[\mathtt{A}]_{i}(x,0)+
\frac{\partial [\mathtt{A}]_{i}}{\partial t}(x,0)t +
\frac{\partial^{2}[\mathtt{A}]_{i}}{{\partial t}^2}(x,0)t ^2+o(t ^2)=\mathcal{O}(t ^2),~i\geq 2.
\end{align}
Therefore the perturbation solution is as follows:
\begin{align}
&[\mathtt{A}]_{PER}(x,t )=[\mathtt{A}]_{0}(x,t )+\lambda[\mathtt{A}]_{1}(x,t )+\sum_{i=2}^{\infty}\lambda^i[\mathtt{A}]_{i}(x,t )\nonumber\\
&=g_{\mathtt{A}}(x)+t  D_\mathtt{A}
g_{\mathtt{A}}''(x)+t  f_\mathtt{A}(x,0)+\lambda(-t  g_{\mathtt{A}}(x)g_{\mathtt{B}}(x)+t  \gamma g_{\mathtt{C}}(x))+o(t )+\mathcal{O}(t^2).
\end{align}
Hence,
$
[\mathtt{A}]_{PER}(x,t )=[\mathtt{A}]_{FDM}(x,t )+\mathcal{O}(t^2).$

\section{Poisson Hypothesis Testing }\label{AppB} 
In a Poisson hypothesis testing problem with two hypotheses, we take a sample from a Poisson random variable $X$ whose mean is either $\rho_0$ (under hypothesis $\textbf{H}_{0}$) or $\rho_1$ (under hypothesis $\textbf{H}_{1}$). 
Assume that the two hypotheses  $\textbf{H}_{0}$ and $\textbf{H}_{1}$ are equally likely  and $\rho_0<\rho_1$. The MAP decision rule compares the observed $X$ with threshold $\mathbb{T}_{h}=({\rho_1-\rho_0})/({\log{\rho_1}-\log{\rho_0}})$ and declares $\textbf{H}_{0}$ if and only if $X<\mathbb{T}_{h}$. The error probability of the MAP decision rule, denoted by $P_e(\rho_0, \rho_1)$, equals
\begin{align} 
\frac{1}{2}\sum_{n\in\mathbb{Z}:~n\geq \mathbb{T}_{h}}^{\infty}\frac{e^{-\rho_0}\rho_{0}^{n}}{n!}+
\frac{1}{2}\sum_{n\in\mathbb{Z}:~0\leq n<\mathbb{T}_{h}}\frac{e^{-\rho_1}\rho_{1}^{n}}{n!}
=\frac{1}{2}-TV(\mathsf{Poisson}(\rho_0),\mathsf{Poisson}(\rho_1))
\end{align}
where 
$TV(\cdot,\cdot))$ is the total variation distance.
 \begin{lem}\label{decrese:error:pro} 
 Fix some $\rho_0$. Then  $P_e(\rho_0,\rho_1)$ is a decreasing continuous function of $\rho_1$ for $\rho_1\geq \rho_0$.
 	\end{lem}
\begin{proof}
The above statement is equivalent with 
$TV(\mathsf{Poisson}(\rho_0),\mathsf{Poisson}(\rho_1))$ being an increasing continuous function of $\rho_1$ for $\rho_1\geq \rho_0$. When $\rho_1=\rho_0$, the total variation distance is zero and as we increase $\rho_1$, this distance increases. Since the distribution of $\mathsf{Poisson}(\rho_1)$ varies continuously, the changes in the total variation distance is also continuous in  $\rho_1$.
\end{proof}

\section{Support Lemma}\label{AppC}
Let $\mathbb{P}$ be the space of all unnormalized probability distributions on the interval $[0,T]$ (\emph{i.e.},  nonnegative functions with finite nonzero integral). For a distribution $p(t)\in\mathbb{P}$ and a continuous function $f(t)$, we define $
\mathbb{E}_{p}(f)=\int_{0}^{T}
p(t)f(t) dt.
$

\begin{lem} \label{SupportLemma} Take arbitrary continuous functions $f(t)$ and $f_i(t)$ for $i=1,2,\cdots, n$, and an arbitrary unnormalized distribution $p\in\mathbb{P}$. Then, there is another \emph{discrete} unnormalized distribution $q\in\mathbb{P}$
taking values 
in a set of size $n$, \emph{i.e.,}
$q(t)=\sum_{i=1}^{n}{a}_{i}\delta(t-t_i)$
for some $a_i\geq 0$ and $t_i\in[0,T]$ 
such that 
\begin{align}\mathbb{E}_{q}(f)\leq \mathbb{E}_{p}(f)\label{eqn:C12}\end{align}
and
\begin{align}\mathbb{E}_{q}(f_i)= \mathbb{E}_{p}(f_i), \quad \text{for}~~i=1,\cdots, n.\label{eqn:Cn}\end{align}
\end{lem}
This lemma shows that by replacing $p$ with $q$, we preserve the $n$ linear constraints \eqref{eqn:Cn} and impose one linear inequality constraint \eqref{eqn:C12}. The support of $q$ (the number of delta functions) is at most the number of constraints which is $n$. 

Support lemmas of this type are commonly used in information theory, and follow from  Fenchel-Bunt's extension of the  Caratheodory theorem. For completeness, we give an intuitive sketch of the proof. For simplicity assume that $p(t)$ is discrete but with support on an arbitrarily large set, \emph{i.e.,}
$p(t)=\sum_{i=1}^{N}g_{i}\delta(t-\tilde{t}_i)$
for some large $N$, and $g_i\geq 0$, $\tilde{t}_i\in[0,T]$. Consider functions of the form
$\tilde{q}(t)=\sum_{i=1}^{N}x_{i}\delta(t-\tilde{t}_i)
$ for some $x_1,\cdots, x_N\geq 0$. Consider the set of $(x_1, \cdots, x_N)$ for which we have $\mathbb{E}_{\tilde q}(f_i)= \mathbb{E}_{p}(f_i)$, for $i=1,\cdots, n$. This imposes $n$ linear constraints on $(x_1, \cdots, x_n)$. These $n$ linear constraints along with the inequality constraints $x_1, \cdots, x_N\geq 0$ define a polytope in the $N$-dimensional  region. This polytope is nonempty since $(x_1, \cdots, x_N)=(g_1, \cdots, g_N)$ belongs to this polytope. To enforce the inequality \eqref{eqn:C12}, let us minimize $\mathbb{E}_{\tilde q}(f)$, which is a linear function in $(x_1, \cdots, x_N)$ over this polytope. The minimum of a linear function occurs at a vertex of the polytope. The key observation is that each vertex of the polytope has at most $n$ nonzero entries,  \emph{i.e.,} if $(x^*_1, \cdots, x^*_N)$ is a vertex of the polytope, at most $n$ entries of $(x^*_1, \cdots, x^*_N)$ are nonzero. This would imply the desired identification for $q(x)$. To see this, observe that since the polytope is in $N$ dimensions, every vertex of the polytope lies at the intersection of $N$ hyperplanes. The hyperplanes defining the polytope are the $n$ linear constraints along with $x_1,\cdots, x_N\geq 0$. Any vertex has to satisfy $N$ of these equations with equality. Thus, the vertex needs to pick at least $N-n$ inequalities of the form $x_i\geq 0$ and satisfy them with equality. In other words, for any vertex, at least $N-n$ entries must be zero, meaning that the number of nonzero entries is at most $n$.

\ifCLASSOPTIONcaptionsoff
  \newpage
\fi

%
%


\section{Further Examples}\label{ex2,3}

\subsection{Example 2: Reaction for Channel Amplification}

 The main example in Section \ref{sec::generalperturbation} used chemical reaction as a means to produce molecules that are detected by the receiver. However, reaction may be used for other purposes as well. For instance, it may be used to enhance the channel between the transmitter and the receiver. This concept is considered in the example below. 

Consider the following example with one transmitter and one receiver. The receiver can only measure the density of molecules of type $\mathtt{A}$ at its location. The transmitter is also able to release molecules of types $\mathtt{A}$ and $\mathtt{B}$ into the medium. Assume that there is an enzyme $\mathtt{C}$ in the medium (outside of our control) which reacts with molecules of type $\mathtt{A}$. If the level of enzyme $\mathtt{C}$ is high, molecules of type $\mathtt{A}$ are mostly dissolved before reaching the receiver. To overcome this, the transmitter may release molecules of a different type $\mathtt{B}$, which would also react with the enzyme $\mathtt{C}$ and thereby reduce the concentration of $\mathtt{C}$ in the medium. This ``cleaning" of the medium from molecules of type $\mathtt{C}$ would enhance the channel from the transmitter to the receiver. More specifically, assume that the medium is governed by the following chemical reactions:
\begin{align}
	\ce{\mathtt{A} + \mathtt{C}&->[\lambda_1]\mathtt{P}_1},\label{eqnNew3}\\
	\ce{\mathtt{B} + \beta  \mathtt{C} &->[\lambda_2] \mathtt{P}_2},\label{eqnNew4}
\end{align}    
where $\mathtt{P}_1$ and $\mathtt{P}_2$ are some products which do not include molecules of type $\mathtt{A},\mathtt{B}$ or $\mathtt{C}$. Here $\beta$ is a natural number. As an example, $\mathtt{A}$ and $\mathtt{B}$ can be two different acids, and $\mathtt{C}$ is a base substance (which reacts with acids $\mathtt{A}$ and $\mathtt{B}$). If $\mathtt{B}$ is a stronger acid than $\mathtt{A}$, the coefficient $\beta$ can be large, and release of $\mathtt{B}$ can be effective in canceling $\mathtt{C}$ from the medium. 

In \eqref{eqnNew3} and \eqref{eqnNew4}, let $\gamma=\lambda_2/\lambda_1$ be the ratio of the reaction rates. 
For simplicity of notation, set  $\lambda_1=\lambda,\lambda_2=\gamma\lambda$.  Assume that the reaction occurs in a two-dimensional medium.  For our modeling purposes, suppose that 
the transmitter is located at the origin and the receiver is located at $(d,0)$.  There is also an independent source that releases molecule  of type $\mathtt{C}$ in the medium with a known concentration $f_{\mathtt{c}}(x,y,t)$ at time $t$. The source is assumed to be located at location $r_0$. See Fig \ref{fig:ex2} for a depiction of this setting. Following equations describe the system dynamic.  
\begin{align}
&\frac{\partial[\mathtt{A}]}{\partial t}=D_\mathtt{A}~\bigtriangledown^2 [\mathtt{A}]-\lambda~[\mathtt{A}]
[\mathtt{C}]+f_{\mathtt{A}}(x,y,t),\label{cartesian2dim:1}\\
&\frac{\partial[\mathtt{C}]}{\partial t}=D_\mathtt{C}~\bigtriangledown^2 [\mathtt{C}]-\lambda~([\mathtt{A}][\mathtt{C}]+\gamma[\mathtt{B}][\mathtt{C}]^{\beta})+f_{\mathtt{C}}(x,y,t),\label{cartesian2dim:2}\\
&	\frac{\partial[\mathtt{B}]}{\partial t}=D_\mathtt{B}~\bigtriangledown^2 [\mathtt{B}]-\gamma\lambda~[\mathtt{B}]
[\mathtt{C}]^{\beta}+f_{\mathtt{B}}(x,y,t),\label{cartesian2dim:3}
\end{align}
where $f_{\mathtt{A}}(x,y,t),f_{\mathtt{B}}(x,y,t)$ are input signals of the transmitter. 
 The  initial conditions for concentration of molecule of type $\mathtt{A}$ and type $\mathtt{B}$ are set to zero.
 For molecule of type $\mathtt{C}$, the initial condition is $[\mathtt{C}](x,y,0)=\mathcal{I}_{\text{int}}(x,y)$. We assume that $\mathcal{I}_{\text{int}}(x,y)$ is completely known. We assume that the diffusion occurs in the entire $\mathbb{R}^2$ and do not assume any boundaries for the medium.
  For the case of no reaction, $\lambda=0$, the system of equations has a closed-form solution.
Consider a solution of the following form:
\begin{align}\label{teylorex:2}
	&[\mathtt{A}](x,y,t)=\sum_{i=0}^{\infty}\lambda^{i}[\mathtt{A}]_{i}(x,y ,t),~
	[\mathtt{C}](x,y,t)=\sum_{i=0}^{\infty}\lambda^{i}[\mathtt{C}]_{i}(x,y ,t),~
	[\mathtt{B}](x,y,t)=\sum_{i=0}^{\infty}\lambda^{i}[\mathtt{B}]_{i}(x,y ,t).
	\end{align}
	 By substituting \eqref{teylorex:2} in 
\eqref{cartesian2dim:1}, \eqref{cartesian2dim:2}, \eqref{cartesian2dim:3}, and 
 matching the coefficients of  $\lambda^0,\lambda^1,\lambda^2$ on both side of equations, we obtain the following equations:
 
 For particle $\mathtt{A}$ we have:
\begin{align}
	&\frac{\partial[\mathtt{A}]_0}{\partial t}=D_\mathtt{A}~\nabla^2 [\mathtt{A}]_0+f_{\mathtt{A}}(x,y,t),~
	\frac{\partial[\mathtt{A}]_1}{\partial t}=D_\mathtt{A}~\nabla^2 [\mathtt{A}]_1-[\mathtt{A}]_0  [\mathtt{C}]_0,\label{eq:ex2:A:01}\\
	&\frac{\partial[\mathtt{A}]_2}{\partial t}=D_\mathtt{A}~\nabla^2 [\mathtt{A}]_2-([\mathtt{A}]_0 [\mathtt{C}]_1+[\mathtt{A}]_1  [\mathtt{C}]_0).\label{eq:ex2:A:2}
\end{align}
\newpage

\begin{figure}[H]
\centering
	\includegraphics[trim={1cm 0cm 0cm 0cm}, scale=.5]{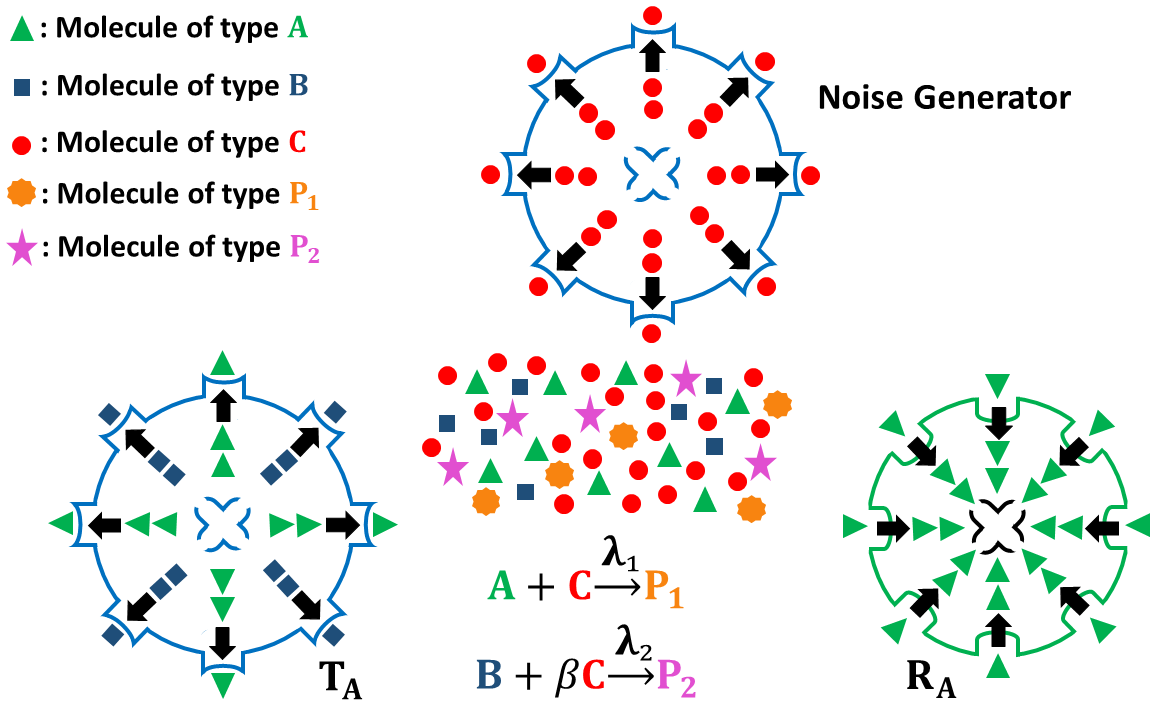}
	\caption{MC system for Example $2$.}
	\label{fig:ex2}
\vspace{-3em}
\end{figure}
\begin{figure}[H]
	\centering
		\includegraphics[trim={1cm 0cm 0cm 0cm}, scale=0.5]{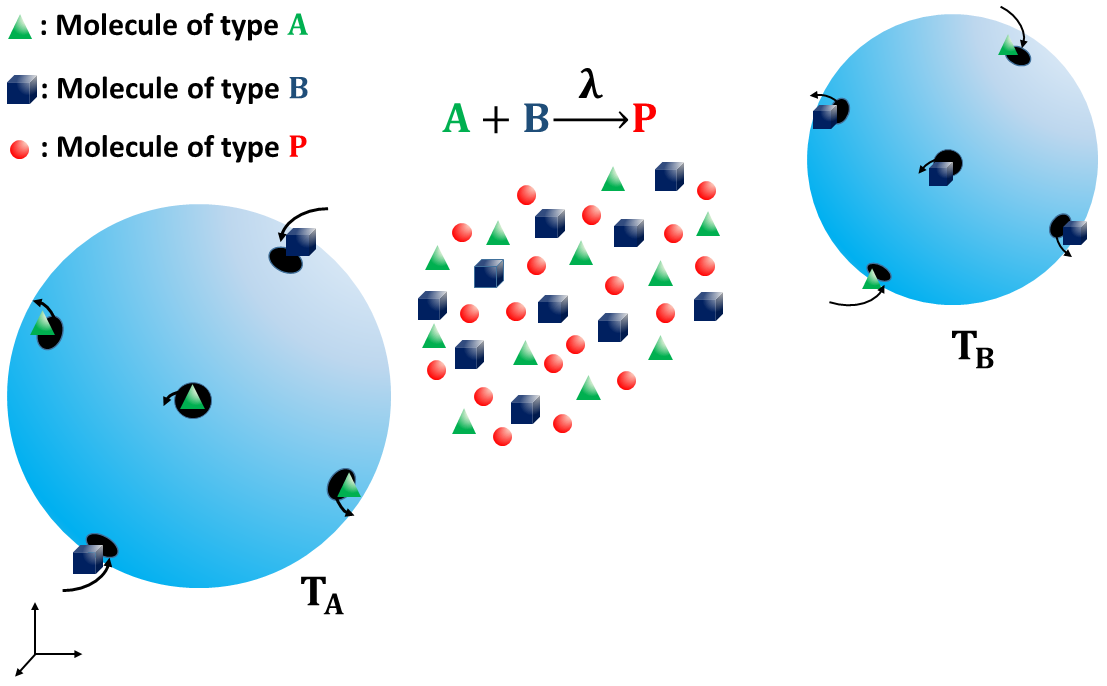}
		\caption{MC system for Example $3$.}
		\label{fig:ex3}
\end{figure}

For molecule of type $\mathtt{C}$ we have: 
\begin{align}\label{eq:ex2:C}
&	\frac{\partial[\mathtt{C}_0]}{\partial t}=D_\mathtt{C}~\nabla^2 [\mathtt{C}]_0+f_{\mathtt{C}}(x,y,t),~~	
	\frac{\partial[\mathtt{C}]_1}{\partial t}=D_\mathtt{C}~\nabla^2 [\mathtt{C}]_1-(\gamma[\mathtt{C}]_{0}^{\beta} [\mathtt{B}]_0+[\mathtt{C}]_{0}[\mathtt{A}]_0).
	\end{align}
To impose the initial condition for concentration of molecule of type $\mathtt{C}$, 
We set $[\mathtt{C}]_{0}(x,y,0)=\mathcal{I}_{\text{int}}(x,y)$ and for $i\geq 1$ we set $[\mathtt{C}]_{i}(x,y,0)=0$.

For particle $\mathtt{B}$ we have:
 \begin{align}\label{eq:ex2:B}
 	&\frac{\partial[\mathtt{B}]_0}{\partial t}=D_\mathtt{B}~\nabla^2 [\mathtt{B}]_0+f_{\mathtt{B}}(x,y,t).
 \end{align}
 Let  $\phi_{\mathtt{A}}(x,y,t)$ be  the impulse response of heat equation with   $D_\mathtt{A}$,
 \begin{align}
 \frac{\partial \phi_\mathtt{A}}{\partial t}=D_\mathtt{A}~\nabla^2(\phi_\mathtt{A})+\delta(x)\delta(y)\delta(t).
 \end{align}
 we have:
 \begin{align}\label{greentwo}
 \phi_\mathtt{A}(x,y,t)=\frac{1}{4\pi D_\mathtt{A} t}\exp(-\frac{x^2+y^2}{4 D_\mathtt{A} t}),~~t\geq 0.
 \end{align}
 Similarly, $\phi_{\mathtt{B}}(x,y,t)$ and $\phi_{\mathtt{C}}(x,y,t)$ are the impulse responses of the heat equations with diffusion coefficients $D_\mathtt{B}$ and $D_\mathtt{C}$ respectively. The solutions of \eqref{eq:ex2:A:01},\eqref{eq:ex2:A:2},\eqref{eq:ex2:C},  and \eqref{eq:ex2:B} are:
\begin{align}
 	&[\mathtt{A}]_0=\phi_\mathtt{A}**f_{\mathtt{A}},~[\mathtt{C}]_0=
 	\phi_\mathtt{C}**(f_{\mathtt{C}}+D_{\mathtt{C}}\nabla^{2}\mathcal{I}_{\text{int}})+\mathcal{I}_{\text{int}},
 	~[\mathtt{B}]_0=\phi_\mathtt{B}**f_{\mathtt{B}},\label{ex2:final:sol:1}\\&[\mathtt{A}]_1=- \phi_\mathtt{A}**([\mathtt{A}]_0 [\mathtt{C}]_0),
~[\mathtt{C}]_1=-\phi_\mathtt{C}**(\gamma[\mathtt{B}]_0 [\mathtt{C}]_0^{\beta}+[\mathtt{C}]_0 [\mathtt{A}]_0),\label{ex2:final:sol:2}\\
 	&[\mathtt{A}]_2=-\phi_{\mathtt{A}}**([\mathtt{A}]_0 [\mathtt{C}]_1+ [\mathtt{A}]_1   [\mathtt{C}]_0).\label{ex2:final:sol:3}
 \end{align}
 The density of molecule $\mathtt{A}$ is equal to:
 \begin{align}\label{ex2:Asol}
 [\mathtt{A}](x,y,t)= [\mathtt{A}]_0+\lambda [\mathtt{A}]_1+\lambda^{2} [\mathtt{A}]_2 +\mathcal{O}(\lambda^3)
 \end{align} 
 For low  reaction rates, we can approximate \eqref{ex2:Asol} as follows:
  \begin{align}\label{ex2:Asol2}
[\mathtt{A}](x,y,t)\approx [\mathtt{A}]_0+\lambda [\mathtt{A}]_1+\lambda^{2} [\mathtt{A}]_2(x,y,t)
 \end{align} 
 Using the above model, we design a modulation scheme in Section \ref{mod:ex2}.
 
 \subsection{Modulation Design For Example 2}\label{mod:ex2}

 Consider a communication scenario consisting of a transmitter and a receiver. The transmitter releases molecules of type $\mathtt{A}$ and $\mathtt{B}$ into the medium to encodes a message $M_{\mathtt{A}}\in\{0,1\}$. The concentration of the released molecules of type $\mathtt{A}$ is described by the input signal $f_{\mathtt{A}}^{i}(x,y,t) =a_i(t)\delta(x)\delta(y)$,$0\le t\le T$ for $i=0,1$. In other words, the density of released molecules of type $\mathtt{A}$ at time $t$ is $a_{M_{\mathtt{A}}}(t)$ if  message $M_{\mathtt{A}}$ is transmitted for $t\in[0,T]$. Similarly,  $f_{\mathtt{B}}(x,y,t)=b_{M_{\mathtt{A}}}(t)\delta(x)\delta(y)$ is the released concentration of molecules of type $\mathtt{B}$ where $b_0(t)$ and $b_1(t)$ are two nonnegative waveforms for $t\in[0,T]$. 
 The total amount of released molecules of type $\mathtt A ,\mathtt B $ during the transmission period $T$ is assumed to be at most $s_{\mathtt A},s_{\mathtt B}$ respectively, \emph{i.e.,}
 \begin{align}\int_{0}^T a_i(t)dt\leq s_{\mathtt A},~\int_{0}^T b_i(t)dt\leq s_{\mathtt B} \quad i=0,1.\label{eqn:power2a}
 \end{align}
 Finally, we assume that molecules of type $\mathtt C$ are being continuously generated throughout the medium  according to some function $f_{\mathtt{C}}(x,y,t)$ for $x,y\in\mathbb{R}$, $t\in[0,T]$. The function  $f_{\mathtt{C}}(x,y,t)$ and the initial density  $[\mathtt{C}](x,y,0)$ at time $t=0$ are assumed to be completely known to the receiver and the transmitter. 
 
 Receiver samples the number of molecules of type $\mathtt{A}$ at $(x=0,y=d,t=T)$. Similar to Example 1,  receiver gets a number from a Poisson distribution with parameter $[\mathtt{A}](0,d,T)$.
 As before, the probability of error is defined as 
 \begin{equation}
 Pr(e)=Pr\{ M_{\mathtt{A}}\neq \hat{M}_\mathtt{A}\},
 \end{equation}
 where $\hat{M}_\mathtt{A}\in\{0,1\}$ is the receiver's decoded message bit. Similar to Example 1, we wish to minimize the error probability by choosing the best possible waveforms $a_i(t)$ and $b_j(t)$.
 
 The following theorem states that in the low reaction rate regime (for the approximate reaction-diffusion equations when we consider  first terms in the Taylor series expansion), one possible  optimal waveforms $a_i(t)$ and $b_i(t)$ that achieve minimum probability of error are as follows:
 \begin{theorem}\label{th::2}
 	For any given noise source $f_{\mathtt{C}}(x,y,t)$, and any arbitrary initial condition $\mathcal{I}_{int}(x,y)$, one  choice for optimal order one waveforms $a_i(t),b_i(t)$ are as follows:
 	\begin{align}
 	&a_0(t)=0,~a_1(t)= s_\mathtt{A}\delta(t-t^{[a_1]}),~b_0(t)=0,~b_1(t)=s_{\mathtt{B}}\delta(t-t^{[b_1]}). 
 	\end{align}
 	for some  		$t^{[a_1]},t^{[b_1]}\in [0,T]$ which depend on
 	$f_{\mathtt{C}}(x,y,t)$, and  $\mathcal{I}_{int}(x,y)$.
 \end{theorem} 
 \begin{proof}\label{AppF}
 	Substituting $f_\mathtt{A}^{i}(x,t),f_\mathtt{B}^{i}(x,t),i=0,1$ in \eqref{recu-eq1}-\eqref{recu-eq3}, one obtains an explicit formula for the (first-order approximation of the) concentration of molecules $\mathtt{A}$ in terms of the input signals. In particular, the concentration $[\mathtt{A}]$ at receiver's location at the sampling time $T$ is as follows:
 	\begin{equation*}\label{rho::ex2}
 	\begin{aligned}
 	[\mathtt{A}](d,0,T)=&[\mathtt{A}]_{0}(d,0,T)+\lambda
 	[\mathtt{A}]_{1}(d,0,T)+\lambda^2[\mathtt{A}]_{2}(d,0,T).\\
 	\end{aligned}
 	\end{equation*}
 	
 	Observe from \eqref{ex2:final:sol:1}, \eqref{ex2:final:sol:2}, and \eqref{ex2:final:sol:3} that $[A]_{0}(x,y,t)$ and $[A]_{1}(x,y,t)$ are linear functions of  $a_{M_{\mathtt{A}}}(t)$ and do not depend on $b_{M_{\mathtt{A}}}(t)$. However, $[A]_{2}(x,y,t)$ depends on both $a_{M_{\mathtt{A}}}(t)$ and  $b_{M_{\mathtt{A}}}(t)$. It is a nonlinear function of $a_{M_{\mathtt{A}}}(t)$. However, if we fix $a_{M_{\mathtt{A}}}(t)$, $[A]_{2}(x,y,t)$ is a linear and positive function in $b_{M_{\mathtt{A}}}(t)$.
 	
 	Let
 	\begin{equation}\label{densitiees::ex2:final}
 	\rho=\left\{ {\begin{array}{*{20}{llll}}
 		\rho_{0}=[\mathtt{A}](d,0,T)& & \text{if}&M_{\mathtt{A}}=0\\
 		\rho_{1}=[\mathtt{A}](d,0,T)& & \text{if}&M_{\mathtt{A}}=1.\\
 		\end{array}}\right.
 	\end{equation}
 	
 	Receiver is assumed to be transparent and due to the particle counting noise, it observes a sample from 
 	$\mathsf{Poisson}(V\rho_{
 		M_{\mathtt{A}}})
 	$
 	for some constant $V$. 
 	Since the effect of molecule $\mathtt{B}$  appears in $\lambda^{2}$ coefficient in the expansion $[\mathtt{A}]=\sum_{=i=0}^{\infty}\lambda^{i}[A]_{i}(x,y,t)$ (i.e, $[\mathtt{A}]_{2}(d,0,T)$), we find optimal waveform of order two for $b_{M_{\mathtt{A}}}(t)$. For $a_{M_{\mathtt{A}}}(t)$  we find optimal waveform of order one. 
 	According to Lemma \ref{decrese:error:pro} in Appendix \ref{AppB}, 
 	the error probability is a decreasing function of 
 	$|\rho_1-\rho_0|$, hence by minimizing  $\rho_0$ and maximizing $\rho_1$ the optimal solution can be obtained. Since $\rho_0\geq 0$, we set $a_0(t),b_0(t)=0$ to minimize $\rho_0$ and we get $\rho_0=0$. Fix some $a_1(t)$, and let $b_{1}^{\text{opt}}(t)$ be a maximizer of $\rho_1$. Since the expression for
 	$\rho_1(a_1(t),b_{1}(t))$, up to its second-order term, is linear in $b_1(t)$ for a fixed value of $a_1(t)$, according to Lemma \ref{SupportLemma} there is waveform of the form $b_{1}^{*}(t)=\hat{b}_{1}\delta(t-t^{[b_1]})$  with a total power \eqref{eqn:power2a} less than or equal to the power of $b_{1}^{\text{opt}}(t)$
 	such that (up to the second-order terms), we have $\rho_1(a_1(t),b_{1}^{\text{opt}}(t))=\rho_1(a_1(t),b_{1}^{*}(t))$. 
 	In other words, replacing $b_{1}^{\text{opt}}(t)$ with $b_{1}^{*}(t)$ would not change the value of $\rho_1$, or the error probability.
 	
 	As
 	$b_{1}^{*}(t)=\hat{b}_{1}\delta(t-t^{[b_1]})$ has a total power constraint $s_{\mathtt{B}}$, we obtain $b_{1}\leq s_{\mathtt{B}}$.
 	Since 
 	$\rho_1(a_1(t),b_{1}^{*}(t))$ is an increasing function in $\hat{b}_{1}$,  we deduce that $\hat{b}_1=s_{\mathtt{B}}$ maximizes $\rho_1$.
 	
 	Next, in order to obtain the first-order optimal waveform for 
 	$a_1(t)$, we recall that $\rho_1$ up to the first-order  is a linear function of  only $a_1(t)$. According to Lemma \ref{SupportLemma} for any arbitrary $a_{1}(t)$ there is some
 	${a}_{1}^{*}(t)=\hat{a}_{1}\delta(t-t^{[a_1]})$ with a total power  less than or equal to the power of $a_1(t)$
 	such that (up to order one terms) $\rho_{1}(a_1(t))=\rho_{1}(a_{1}^{*}(t))$. Then, as above the power constraint implies that $\hat{a}_{1}\leq s_{\mathtt{A}}$.  Since (up to order one terms), $\rho_{1}(a_{1}^{*}(t))$ is an increasing function with respect to $\hat{a}_{1}$
 	we deduce $\hat{a}_{1}=s_{\mathtt{A}}$. That completes the proof.
 \end{proof}	
 In order to determine parameters in the statement of Theorem \ref{th::2} we need optimize over constants $t^{[a_1]},t^{[b_1]}$. 
 As an example, consider  the values of system parameters as follows in Table \ref{tabex2}.
 \begin{table}
  	\centering
  	\caption{Parameters}
	 \label{tabex2}
\begin{tabular}{ |c|c|c| } 

	\hline
	\multirow{2}{*}{$(D_{\mathtt{A}},D_{\mathtt{B}},D_{\mathtt{C}})[m^{2}s^{-1}]$}   &\multirow{2}{*} {$10^{-10}\times(10,1.1,1)$}\\
	&\\
	\hline 
	\multirow{2}{*} {$\lambda[molecules^{-1}.m^{3}.s^{-1}]$}   & \multirow{2}{*}{$ 10^{-23}$}\\
	&\\
	\hline
	\multirow{2}{*}{$V[m^{3}]$}   &\multirow{2}{*} {$2.5\times10^{-7}$}\\
	&\\ 
	\hline
	\multirow{2}{*}{$(T[s],\beta,\gamma)$ }& \multirow{2}{*}{$(10,2,1)$}  \\ 
	&\\
	\hline
	
	\multirow{2}{*}{$(d_{\mathtt{B}} [m],d_{\mathtt{c}} [m])$} & \multirow{2}{*}{$5\times 10^{-5}\times((1,1),(10,0))$ } \\
	&\\ 
	\hline
	\multirow{2}{*}{$(\mathcal{I}_{int}(x,y),f_{mathtt{c}}(x,y,t))$} & \multirow{2}{*}{$(0,4\times 10^{4}\delta(x-5\times 10^{5})\delta(y-5\times 10^{5})\delta(t))$} \\
	&\\
		\hline
	\multirow{2}{*}{$s_{\mathtt{A}}[molecules.m^{-3}],s_{\mathtt{B}}[molecules.m^{-3}])$} & \multirow{2}{*}{$10^8\times(5,24)$} \\
	&\\
	\hline 
\end{tabular}
\end{table}
 Simulation results yield the optimal values for 
 $t^{[a_1]},t^{[b_1]}$ as $t^{[a_1]}=5.62,t^{[b_1]}=0$.
 Also, when we vary $s_{\mathtt A}\in [500,5\times 10^{11}]$ while we fix the other parameters, we observe that $t^{[b_1]}$ remains zero. In other words, it is best to release molecules of type $\mathtt B$ as soon as interfering molecules of type $\mathtt C$ are released. Moreover, the value of
 $t^{[a_1]}$ is almost constant as we vary $s_{\mathtt A}$. 

 \subsection{Example 3: Reaction for Two-way Communication }
 Finally, our third example considers two transceivers (who are able to both transmit and receive signals) in a three-dimensional setting.  
Consider two molecular transceivers $\mathsf T_1$ and $\mathsf T_2$. The  transceiver $\mathsf T_1$ is able to release molecules of type $\mathtt{A}$ and receive molecules of type $\mathtt{B}$. On the other hand, the  transceiver $\mathsf T_2$ is able to release molecules of type $\mathtt{B}$ and receive molecules of type $\mathtt{A}$. If there is no reaction between $\mathtt{A}$ and $\mathtt{B}$, we have two distinct  directional channels between the two transceivers (e.g., one channel is formed by $\mathsf T_1$ releasing molecules of type $\mathtt{A}$ and $\mathsf T_2$ receiving them). However, if $\mathtt{A}$ reacts with $\mathtt{B}$ in the medium, the two channels become entangled.  This would impact the transmission strategy of the  two transceivers if they wish to establish a two-way communication channel, and send and receive messages at the same time. The chemical reaction would weaken both signals  in this case. Assume that
\begin{equation}
\ce{\mathtt{A}  + \mathtt{B} ->[\lambda] \mathtt{P}}.
\end{equation}
 The following equations describe the dynamic of the system:
 \begin{align}
	&\frac{\partial[\mathtt{A}]}{\partial t}=D_\mathtt{A}\nabla^2[\mathtt{A}]-\lambda~[\mathtt{A}][\mathtt{B}]+f_\mathtt{A}(\vec{x},t),~\frac{\partial[\mathtt{B}]}{\partial t}=D_\mathtt{B}\nabla^2[\mathtt{B}]-\lambda~[\mathtt{A}][\mathtt{B}]+f_\mathtt{B}(\vec{x},t),\label{ex3:eq}
\end{align}
where $f_\mathtt{A}(\vec{x},t),f_\mathtt{B}(\vec{x},t)$ are input signals of the two transceivers. The  initial conditions are set to zero, and we allow diffusion in the entire $\mathbb{R}^3$ with no boundaries. As before, for 
$\lambda=0$ the system of equations is linear and analytically solvable.
Consider a solution of the following form:
\begin{align}\label{teylorex:2N}
&[\mathtt{A}](\vec{x},t)=\sum_{i=0}^{\infty}\lambda^{i}[\mathtt{A}]_{i}(\vec{x} ,t),~[\mathtt{B}](\vec{x},t)=\sum_{i=0}^{\infty}\lambda^{i}[\mathtt{B}]_{i}(\vec{x} ,t).
\end{align}
By matching the coefficients of $\lambda^0,\lambda^1$ on the both side of equations, we can find a solution for in low reaction rate regime.
By matching the constant terms, we obtain
\begin{align}\label{ex3:firstterm}
	&\frac{\partial[\mathtt{A}]_0}{\partial t}=D_\mathtt{A}
	\nabla^{2}[\mathtt{A}]_0 +f_\mathtt{A}(\vec{x},t),~\frac{\partial[\mathtt{B}]_0}{\partial t}=D_\mathtt{B}\nabla^{2}[\mathtt{B}]_0 +f_\mathtt{B}(\vec{x},t).
\end{align}
By matching the coefficients of $\lambda$, we get
\begin{align}\label{ex3:secondterm}
	&\frac{\partial[\mathtt{A}]_1}{\partial t}=D_\mathtt{A}
	\nabla^{2}[\mathtt{A}]_1 -[\mathtt{A}]_{0}[\mathtt{B}]_{0},~\frac{\partial[\mathtt{B}]_1}{\partial t}=D_\mathtt{B}\nabla^{2}[\mathtt{B}]_1 -[\mathtt{A}]_{0}[\mathtt{B}]_{0}.
\end{align} 
The impulse response for heat equation is as follows:
\begin{equation}\label{greenthree}
\phi_\mathtt{A}(\vec{x},t)=\frac{1}{(4\pi D_\mathtt{A}t)^{\frac{3}{2}} }\exp(-\frac{\lVert\vec{x}\rVert_{2}^{2}}{4 D_\mathtt{A} t}),  ~~\vec{x}\in \mathbb{R}^{3}~,t\geq 0.
\end{equation}
 The solution of \eqref{ex3:firstterm},
 \eqref{ex3:secondterm} are given in the following form.
\begin{align}\label{ex:3:final:density}
	&[\mathtt{A}]_0=\phi**f_\mathtt{A},~
	[\mathtt{B}]_0=\phi**f_\mathtt{B},~
	[\mathtt{A}]_1=- \phi_{\mathtt{A}}**([\mathtt{A}]_0 [\mathtt{B}]_0),~
    [\mathtt{B}]_1=- \phi_{\mathtt{B}}**([\mathtt{A}]_0 [\mathtt{B}]_0).
	\end{align}
This results in the following solution for low reaction rates:
 \begin{align}\label{ex3::finalsol}
 	 &[\mathtt{A}](\vec{x},t)= [\mathtt{A}]_0(\vec{x},t)+\lambda [\mathtt{A}]_1(\vec{x}, ,t)+\mathcal{O}(\lambda^2),~[\mathtt{B}](\vec{x},t)= [\mathtt{B}]_0(\vec{x},t)+\lambda [\mathtt{B}]_1(\vec{x}, ,t)+\mathcal{O}(\lambda^2).
 	 \end{align}
 Using these equations, we design a modulation scheme in Section \ref{mod:ex3}.
 \subsection{Design of Modulation For Example 3}\label{mod:ex3}
 Suppose transceivers $ \mathsf T_{\mathtt {A}}, \mathsf T_{\mathtt {B}}$ encode messages $M_{\mathtt{A}},M_{\mathtt{B}}\in\{0,1\}$ with input signals $f_{\mathtt{A}}^{i}(\vec{x},t) =a_i(t)\delta(\vec{x}),f_{\mathtt{B}}^{i}(\vec{x},t)  =b_i(t)\delta(\vec{x}-\vec{d_{\mathtt{B}}})$, for $0\le t\le T$ and $i=0,1$,
 where $\vec{d_\mathtt{B}}=(d,0,0)
 $ is the location of $\mathsf T_{\mathtt {B}}$. In other words, transceiver $ \mathsf T_{\mathtt {A}}$ releases a concentration of $a_{M_{\mathtt{A}}}(t)$ at origin, and 
 transceiver $ \mathsf T_{\mathtt {B}}$ releases a concentration of $b_{M_{\mathtt{B}}}(t)$ at its location $\vec{d_\mathtt{B}}$. As before, we restrict the total amount of released molecules of types $\mathtt A$ and $\mathtt B$ during the transmission period $T$ as follows:
 \begin{align}\int_{0}^T a_i(t)dt\leq s_{\mathtt A},~\int_{0}^T b_j(t)dt\leq s_{\mathtt B} \quad i=0,1.\label{eqn:power3a}
 \end{align}
 
 Transceiver $\mathsf{T_{\mathtt {A}}}$ samples the medium for molecules of type $\mathtt {A}$ at the end of the time slot at time $t=T$, and uses its observation to decode the message $\hat{M}_{\mathtt{B}}$. Similarly, transceiver 
 $\mathsf{T_{\mathtt {B}}}$ decodes the message $\hat{M}_{\mathtt{A}}$ using its sample at time $T$.
 Four error probabilities could be considered in our problem: since transceiver $ \mathsf T_{\mathtt {B}}$ knows transmitted message $M_{\mathtt{B}}$, for this transceiver we can consider error probabilities 
 \begin{align*}&J_1:= Pr(M_{\mathtt{A}}\neq \hat M_{\mathtt{A}}|M_{\mathtt{B}}=0),
 ~J_2:= Pr(M_{\mathtt{A}}\neq \hat M_{\mathtt{A}}|M_{\mathtt{B}}=1)\end{align*}
 Similarly, for transceiver $ \mathsf T_{\mathtt {A}}$, we can consider 
 \begin{align*}&J_3:= Pr(M_{\mathtt{B}}\neq \hat M_{\mathtt{B}}|M_{\mathtt{A}}=0),~
 J_4:= Pr(M_{\mathtt{B}}\neq \hat M_{\mathtt{B}}|M_{\mathtt{A}}=1)\end{align*}
 We would like to make  $J_1, J_2, J_3, J_4$ as small as possible 
 by choosing optimal waveforms $a_i(t)$ and $b_j(t)$. Since there is a tradeoff between minimizing $J_1, J_2, J_3$ and $J_4$, we can choose a cost function
 $\mathcal{H}:[0,1]^{4}\rightarrow \mathbb{R}$ 
 and minimize 
 $\mathcal{H}(J_{1},J_{2},J_{3},J_{4})$. While we could choose any arbitrary cost function, for simplicity of exposition, we adopt $\mathcal{H}(J_1,J_2,J_3,J_4)=\sum_{i=1}^{4}\omega_{i}\log(J_{i})
 $ for some constants $\omega_{i}, i=1, \cdots,4$. This particular choice for $\mathcal{H}$ leads to very simple  optimal waveforms $a_i(t)$ and $b_j(t)$ in the low reaction regime as shown in the following theorem (other choices for $\mathcal{H}$ result in more delta terms in the optimal waveform). 
 \begin{theorem}\label{th::3}
 	One  choice for optimal order one waveforms $a_i(t), b_i(t)$ is as follows:
 	\begin{align}
 	&a_0(t)=0,~a_1(t)=\hat{a}_{1}\delta(t-t^{[a_1]}),~b_0(t)=0,~b_1(t)=\hat{b}_{1}\delta(t-t^{[b_1]}).
 	\end{align}
 	for some $\hat{a}_{1},\hat{b}_{1},t^{[a_1]},t^{[b_1]}$. 
 \end{theorem}
 \begin{proof}
 	By substituting $f_{\mathtt{A}}^{0}(\vec{x},t),f_{\mathtt{A}}^{1}(\vec{x},t),f_{\mathtt{B}}^{0}(\vec{x},t)$ and  $f_{\mathtt{B}}^{1}
 	(\vec{x},t)$ in \eqref{ex:3:final:density},
 	\eqref{ex3::finalsol}, the density of molecule $\mathtt{A}$ at  $\vec{d}_{\mathtt{B}}$ and
 	density of molecule $\mathtt{B}$ at $\vec{d}_{\mathtt{A}}=\vec{0}$  have an explicit formula with respect to input signals.
 	The concentration  of molecule of type $\mathtt{A}$
 	at location of $\mathsf{T_{\mathtt {B}}}$  at the sampling time $T$ is $	[\mathtt{A}](\vec{d}_
 	{\mathtt{B}},T)=\zeta_1-\zeta_2,$ where:
 	\begin{align}
 	\zeta_1=\int_{0}^{T}
 	\phi_{\mathtt{A}}(\vec{d}_{\mathtt{B}},T-t^{'})a_{M_\mathtt{A}}(t^{'})dt^{'},
 	\end{align}
 	and
 	\begin{align}
 	\zeta_2=
 	\lambda
 	\int_{0}^{T}
 	\int_{\mathbb{R}^3}\int_{0}^{t'}
 	\int_{0}^{t'}
 	&\phi_{\mathtt{A}}(\vec{d}_{\mathtt{B}}-\vec{x}',T-t')\phi_{\mathtt{A}}(\vec{x}',t'-t_1) 
 	\phi_{\mathtt{B}}(\vec{x}'-\vec{d}_{\mathtt{B}},t'-t_2)
 	\nonumber\\&a_{M_{\mathtt{A}}}(t_1)
 	b_{M_{\mathtt{B}}}(t_2)
 	dt_{1}dt_{2}d\vec{x}'dt'.
 	\end{align}
 	A similar expression holds for the concentration of molecule $[\mathtt B]$ at the location of  $\mathsf{T_{\mathtt {A}}}$ at sampling time $T$ (by simply swapping $\mathtt{A}\rightleftharpoons\mathtt{B}$ in the above formula).
 	Observe that $[\mathtt{A}](\vec{d}_{\mathtt{B}},T)$ and $[\mathtt{B}](\vec{0},T)$ are bilinear functions with respect to input signals. 
 	
 	Let
 	$\rho_{m_{\mathtt{B}}=0}^{\mathtt{A}}$ be the concentration of molecules of type $\mathtt{A}$ at transceiver $\mathtt{B}$ if $m_{\mathtt{B}}=0$. Similarly, we define $\rho_{m_{\mathtt{B}}=1}^{\mathtt{A}}$. Also, 
 	$\rho_{m_{\mathtt{A}}=1}^{\mathtt{B}}$ denotes the concentration of molecules of type $\mathtt{B}$ at transceiver $\mathtt{A}$ if $m_{\mathtt{A}}=0$, etc.
 	
 	Transceivers are assumed to be transparent. Due to the counting noise transceiver $\mathtt{A}$  observes a sample from $\mathsf{Poisson}(V\rho^{\mathtt{A}}_{m_{\mathtt{B}}})$
 	for some constant $V$. A similar statement holds for transceiver $\mathtt{B}$.

 	According to Lemma \ref{decrese:error:pro} error probabilities $J_1,J_2,J_3,J_4$ are  decreasing functions of 
 	$|\rho_{01}^{\mathtt{B}}-\rho_{00}^{\mathtt{B}}|,
 	|\rho_{11}^{\mathtt{B}}-\rho_{10}^{\mathtt{B}}|,
 	|\rho_{01}^{\mathtt{A}}-\rho_{00}^{\mathtt{A}}|,
 	|\rho_{11}^{\mathtt{A}}-\rho_{10}^{\mathtt{A}}|$ respectively. By minimizing non negative numbers $\rho_{00}^{\mathtt{B}}, \rho_{10}^{\mathtt{B}}$,
 	$\rho_{00}^{\mathtt{A}},\rho_{10}^{\mathtt{A}}$ the probabilities reduce. We set $a_0(t),b_0(t)=0$ to achieve $(\rho_{00}^{\mathtt{B}}, \rho_{10}^{\mathtt{B}},
 	\rho_{00}^{\mathtt{A}},\rho_{10}^{\mathtt{A}})=(0,0,0,0)$. To obtain $a_1(t),b_1(t)$, we consider cost function  $\mathcal{H}(J_1,J_2,J_3,J_4)=\sum_{i=1}^{4}\omega_{i}\log(J_{i})$ for some constants $\omega_{i}, i=1, \cdots,4$. 
 	The expression for $\mathcal{H}(J_1,J_2,J_3,J_4)$ becomes a bilinear function of $a_1(t),b_1(t)$ since  
 	\begin{align}
 	\mathcal{H}(J_1,J_2,J_3,J_4)=&-\omega_{1}\zeta_{1}(\vec{d}_{\mathtt{B}},a_{1})-\omega_{2}\zeta_{2}(\vec{d}_{\mathtt{B}},a_{1},b_{1})-\omega_{3}\zeta_{1}(\vec{d}_{\mathtt{A}},b_{1})-\omega_{4}\zeta_{2}(\vec{d}_{\mathtt{A}},a_{1},b_{1})\nonumber\\&-\sum_{i=1}^{4}\omega_{i}\log(2).
 	\end{align}
 	Fix some  $b_{1}(t)$ and let $a_{1}^{\text{opt}}(t)$ be maximizer of $\mathcal{H}$, according to Lemma \ref{SupportLemma} there is a waveform of the form 
 	${a}_{1}^{*}(t)=\hat{a}_{1}\delta(t-t^{[a_1]})$
 	with a total power \eqref{eqn:power3a} less than or equal to $a_{1}^{\text{opt}}(t)$ such that it dos
 	not affect the value of $\mathcal{H}$, \emph{i.e.}, $\mathcal{H}(a_{1}^{\text{opt}}(t),b_{1}(t))=\mathcal{H}(a_{1}^{*}(t),b_{1}(t))$. As ${a}_{1}^{*}(t)=\hat{a}_1\delta(t-t^{[a_1]})$ has a total power constraint $s_{\mathtt{A}}$, we obtain $\hat{a}_{1}\leq s_{\mathtt{A}}$. Hence    ${a}_{1}^{\text{opt}}(t)=\hat{a}_1\delta(t-t^{[a_1]})$. By similar argument  we can deduce that the one choice for optimal waveform $b_{1}(t)$ is 
 	${b}_{1}^{*}(t)=\hat{b}_1\delta(t-t^{[b_1]})$  for some constant $\hat{b}_1\leq s_{\mathtt{B}},t^{[b_1]}\leq T$ . That completes the proof.
 \end{proof}  
 In order to determine parameters in the statement of Theorem \ref{th::3} we need optimize over constants $\hat{a}_{1},\hat{b}_{1},t^{[a_1]},t^{[b_1]}$.
 \begin{table}
  	\centering
  	\caption{Parameters}
	 \label{tabex3}
\begin{tabular}{ |c|c|c| } 

	\hline
	\multirow{2}{*}{$(D_{\mathtt{A}},D_{\mathtt{B}})[m^{2}s^{-1}]$}   &\multirow{2}{*} {$10^{-10}\times(10,1.1)$}\\
	&\\
	\hline 
	\multirow{2}{*} {$\lambda[molecules^{-1}.m^{3}.s^{-1}]$}   & \multirow{2}{*}{$ 10^{-30}$}\\
	&\\
	\hline
	\multirow{2}{*}{$V[m^{3}]$}   &\multirow{2}{*} {$2.5\times10^{-14}$}\\
	&\\ 
	\hline
	\multirow{2}{*}{$T[s]$ }& \multirow{2}{*}{$10$}  \\ 
	&\\
	\hline
	
	\multirow{2}{*}{$d_{\mathtt{B}} [m]$} & \multirow{2}{*}{$10^{-4}\times(1,1,1)$ } \\
	&\\
		\hline
	\multirow{2}{*}{$(s_{\mathtt{A}}[molecules.m^{-3}],s_{\mathtt{B}}[molecules.m^{-3}])$} & \multirow{2}{*}{$10^8\times(5,24)$} \\
	&\\
	\hline 
\end{tabular}
\end{table}
Simulation results yield the optimal values for $\hat{a}_{1},\hat{b}_{1},t^{[a_1]},t^{[b_1]}$ as $\hat{a}_{1}=2\times 10^{8} ,\hat{b}_{1}=2.4\times 10^{9},t^{[a_1]}=0,$ and $t^{[b_1]}=0$. 

\bibliographystyle{IEEEtran}
\bibliography{Citations}
 
\end{document}